\documentclass{article}

\usepackage{amsmath,bbm}  
\usepackage{natbib}
\usepackage{etoolbox} 

\usepackage{forest}

\usepackage{pythonhighlight}
\usepackage{inconsolata}

\usepackage{enumitem}
\usepackage{array}
\usepackage{float}
\usepackage[super]{nth}
\usepackage{bbm}

\usepackage{authblk}
\usepackage{amsfonts}
\usepackage{amssymb}
\usepackage{amsthm}
\usepackage{mathtools}
\usepackage{graphicx}
\usepackage[hidelinks]{hyperref}

\newtheorem{lemma}{Lemma}
\newtheorem{prop}{Proposition}
\newtheorem{thm}{Theorem}
\newtheorem{cor}{Corollary}

\newtheorem{defn}{Definition}

\newtheorem{example}{Example}

\DeclareFontFamily{U}{anothermathbb}{}
\DeclareFontShape{U}{anothermathbb}{m}{n}{<-> fourier-bb}{}
\DeclareSymbolFont{anothermathbb}{U}{anothermathbb}{m}{n}
\DeclareSymbolFontAlphabet{\anothermathbb}{anothermathbb}

\DeclareMathOperator*{\argmax}{arg\,max}

\title{Generalizing matrix representations to fully heterochronous ranked tree shapes}

\author[1,6]{Chris Jennings-Shaffer}
\author[2]{Ziyue (Cherith) Chen}
\author[2,3]{Julia A Palacios\textsuperscript{*}}
\author[1,4,5,6]{Frederick A Matsen IV\textsuperscript{*}}

\affil[1]{Fred Hutchinson Cancer Research Center, Seattle, Washington, USA}
\affil[2]{Department of Statistics, Stanford University, Stanford, CA}
\affil[3]{Department of Biomedical Data Science, Stanford School of Medicine, Stanford, CA}
\affil[4]{Department of Genome Sciences, University of Washington, Seattle, WA}
\affil[5]{Department of Statistics, University of Washington, Seattle, WA}
\affil[6]{Howard Hughes Medical Institute, Seattle, WA}

\affil[*]{Co-corresponding authors: \texttt{juliapr@stanford.edu} and \texttt{matsen@fredhutch.org}}

\date{}

\begin{document}
\maketitle

\begin{abstract}
  Phylogenetic tree shapes capture fundamental signatures of evolution.
  We consider ``ranked'' tree shapes, which are equipped with a total order on the internal nodes compatible with the tree graph.
  Recent work has established an elegant bijection between ranked tree shapes and a class of integer matrices, called \textbf{F}-matrices, defined by simple inequalities.
  This formulation is for isochronous ranked tree shapes, where all leaves share the same sampling time, such as in the study of ancient human demography from present-day individuals.
  However, branch lengths of phylogenetic trees can represent units other than calendar time, such as evolutionary distance.
  A tree equipped with branch lengths quantifying evolutionary distance, called a rooted phylogram, is output by popular maximum-likelihood methods.
  These trees are broadly relevant, such as to study the affinity maturation of B cells in the immune system.
  Discretizing time in a rooted phylogram gives a fully heterochronous ranked tree shape, where leaves are part of the total order.
  Here we extend the \textbf{F}-matrix framework to such fully heterochronous ranked tree shapes.
  We establish an explicit bijection between a class of \textbf{F}-matrices and the space of such tree shapes.
  The matrix representation has the key feature that the value at any entry is highly constrained by four previous entries, enabling straightforward enumeration of all valid tree shapes.
  We also use this framework to develop probabilistic models on ranked tree shapes.
  Our work extends understanding of combinatorial objects that have a rich history in the literature.
\end{abstract}

\section*{Introduction}

Evolution is the unifying theme of biology, and it operates in diverse modes.
These modes can be seen in the structure of phylogenetic trees~\citep{MooersHeard}.
For example, the tree of influenza has a highly ``imbalanced'' shape, which comes from intense evolutionary selective pressure from host immunity, in contrast with the trees of other viruses~\cite{Grenfell2004-dz}.
Scientists characterize these modes of evolution by studying phylogenetic tree ``shapes'': rooted bifurcating tree graphs without leaf labels.

An elegant means of characterizing tree shapes has recently been developed that includes information about the relative ordering of nodes in addition to graph structure~\cite{Kim2020-ip,Samyak2024-yo}.
This relative ordering is expressed as a ranking, i.e., a total ordering of the internal nodes of the tree.
The combination of tree shape and relative ordering defines a ``ranked tree shape.''
There is a bijection between such ranked tree shapes and a class of integer-valued matrices, called ``\textbf{F}-matrices'', characterized by simply-expressed inequalities~\cite{Kim2020-ip,Samyak2024-yo}.
By recording information about the order of events on the tree, this formulation enables richer comparison than tree structure alone.
However, the existing formulation of \textbf{F}-matrices is limited to ``isochronous'' ranked tree shapes (Figure~\ref{fig:treeTypes}, left) in which all the leaves of the tree are assumed to have been sampled at the same time, or at least at known fixed sampling times.
This makes perfect sense in the setting of ``time trees'' (a.k.a. chronograms): phylogenetic trees with nodes labeled by calendar time and leaf nodes representing molecular sequences with known sampling times.
Such trees result from inference done using software such as BEAST~\cite{Drummond2007-co, Bouckaert2019-wv, Baele2025-fe} or TreeTime~\cite{Sagulenko2018-xl}.

\begin{figure}[h!]
  \centering
  \includegraphics[width=0.9\textwidth]{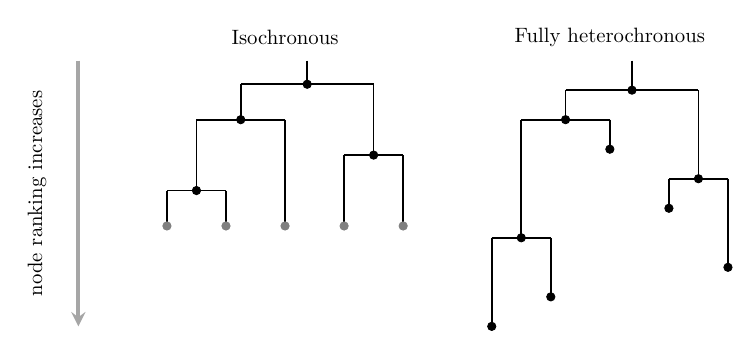}
  \caption{%
    Left: an isochronous tree shape; an \textbf{F}-matrix bijection has been established for such objects \cite{Kim2020-ip,Samyak2024-yo}.
    Right: a fully heterochronous tree shape; the present manuscript establishes an analogous bijection between these objects and a class of \textbf{F}-matrices.
    The two trees are isomorphic as graphs, but are not the same type of ranked tree shape.
    On the left isochronous tree, internal nodes have unique ranks and leaves share a common rank.
    We mark leaves in gray to indicate that they do not form part of the ``data'' encoded by the ranked tree.
    On the right fully heterochronous tree, all nodes have unique ranks and the rank of a leaf may be less than the rank of an internal node.
  }
  \label{fig:treeTypes}
\end{figure}

There is another type of tree analysis that simply represents the phylogenetic tree without timing constraints, letting the length of each edge represent the amount of evolution that has happened along that edge.
This structure is called a ``rooted phylogram''.
Phylograms are the inferential output of software such as IQ-TREE~\cite{Minh2020-va} and RAxML~\cite{Stamatakis2014-xo}.

One may wish to use rooted phylograms to study patterns of evolution in systems where dates are not available or relevant.
For example, in B cell affinity maturation, the evolutionary structure of the phylogenetic tree is determined by a relatively short period in the germinal center, after which the resulting cells circulate for longer as memory B cells without further mutation~\cite{Victora2022-fm}.
Due to this two-part process, the blood sampling time for B cells is not relevant to the actual ``sampling time'' of the B cells, which corresponds to the various times when they left the germinal center.
Hence, one can use a rooted phylogram.

Crucially, the leaf positions of rooted phylogram inference form part of the inferential \emph{output}, in contrast to time tree inference (for which they form part of the \emph{input} data).
Thus, we wish to capture the positions of the leaf nodes as part of our tree representation.
Consequently, we discretize the positions of all nodes into a total ranking, obtaining what we call a \emph{fully heterochronous ranked tree shape} (Figure~\ref{fig:treeTypes}, right).

In this paper, we extend the previous \textbf{F}-matrix characterization to fully heterochronous ranked tree shapes and additionally prove theorems for matrix construction.
Taking this further, we provide a method to iteratively construct all \textbf{F}-matrices one entry at a time.
This extends previous work in two ways.
First, in \citep{Kim2020-ip}, the authors proposed a bijective \textbf{F}-matrix encoding of ranked tree shapes (isochronous and heterochronous) to define a distance on the space via matrix norms.
This defined \textbf{F}-matrices by the ranked tree shapes that they encode, while the set of matrices comprising \textbf{F}-matrices was explicitly identified for isochronous ranked tree shapes only \citep{Samyak2024-yo}; we restate the latter result as Theorem \ref{theorem.isochronous_f_matrix_tree}.
Second, an iterative construction was noted for the isochronous case, but it was not explicitly stated nor was its correctness proved.
As we will show, in both the isochronous and heterochronous case, \textbf{F}-matrices are classified by their entries satisfying a number of simple linear inequalities.
The proposed iterative matrix construction is a method to solve the linear inequalities without the need for back-substitution.
This construction also yields an explicit enumeration of all \textbf{F}-matrices, and as such all fully heterochronous ranked tree shapes. Additionally, this construction facilitates the definition of novel probabilistic models.

Probability distributions on trees are important in phylogenetics, such as for maximum likelihood estimation and as priors in Bayesian inference. 
Examples of widely used probability distributions on isochronous phylogenetic trees include Kingman's coalescent model on ranked labeled tree shapes and Tajima's coalescent model on ranked unlabeled tree shape \citep{kingman1982coalescent, Tajima1983}. 
The current literature lacks descriptive probability distributions on the space of fully heterochronous ranked tree shapes beyond birth-death models that assume a uniform tree topology. 
To address this, we first introduce two parameter-free models: a backward-in-time coalescent model \citep{kingman1982coalescent}, and a forward-in-time model referred to here as diagonal top-down.
These may be considered as null models.
In the opposite direction, we exploit the iterative matrix construction to define highly flexible probability distributions with many free parameters on the space of fully heterochronous ranked tree shapes.
This general construction can be specialized to a particular class of beta-splitting model \citep{aldous1996probability,sainudiin2016beta}.
Future work will focus on fitting these flexible distributions via neural networks.

The remainder of this article is structured as follows.
In Section \ref{section.known} we provide definitions and review connections with the previous literature.
In Section \ref{section.preliminaries} we introduce and provide examples of the types of matrices used here.
In Section \ref{section.theorems} we state and prove theorems for various bijections, classify \textbf{F}-matrices by constraints on their entries, and constructively enumerate \textbf{F}-matrices.
In Section \ref{section.enumeration} we describe two null distributions based on simple sampling schemes for fully heterochronous ranked tree shapes, define a highly flexible non-parametric family of probability distributions on \textbf{F}-matrices, and specialize this to a novel two-parameter family of distributions on ranked tree shapes.
Lastly, in Section \ref{section.discussion} we give a brief discussion of results and directions for the future.

\section{Definitions and connection with previous literature}\label{section.known}
We begin by more formally defining terms and providing connections with previous mathematical literature.
A \emph{fully heterochronous ranked tree shape} is a rooted full binary tree with a total ordering on the nodes such that nodes appear in increasing order along any path from root to leaf (see Example \ref{example.heterochronous}).
Nodes represent events in time and the total ordering is based on time, so no two events (including the sampling of leaves) occur at the same time, hence the term ``fully heterochronous''.
In contrast, an \emph{isochronous ranked tree shape} is a rooted full binary tree with a total ordering on only the internal nodes, but again internal nodes appear in increasing order along any path from root to leaf (see Example \ref{example.isochronous}).
With nodes representing events in time, isochronous ranked tree shapes correspond to different times for all internal nodes and the same time for all leaves.
While one can consider heterochronous ranked tree shapes, where some intermediate number of leaves share ranks, we do not do so in this article.
Table \ref{tab:summary1} shows a summary describing the two types of ranked trees.

\begin{table}[H]
  \small
  \centering
  \begingroup
  \begin{tabular}{|>{\raggedright\arraybackslash}p{0.45\textwidth}|
    >{\raggedright\arraybackslash}p{0.45\textwidth}|}
    \hline
    \textbf{Isochronous} & \textbf{Fully heterochronous } \\
    \hline
    \vspace{-5pt}
    \begin{itemize}[label=-,leftmargin=1.2em,topsep=0pt,partopsep=0pt,parsep=0pt,itemsep=0pt]
      \item Discretized inferential output of e.g.\ BEAST or TreeTime.
      \item Branch lengths are in units of calendar time.
      \item Internal nodes are totally ordered.
      \item Leaf positions are assumed to be equal and known.
    \end{itemize}
    &
    \vspace{-5pt}
    \begin{itemize}[label=-,leftmargin=1.2em,topsep=0pt,partopsep=0pt,parsep=0pt,itemsep=0pt]
      \item Discretized inferential output of e.g.\ IQ-TREE or RAxML.
      \item Branch lengths are in units of evolutionary change.
      \item All nodes are totally ordered.
      \item The leaf positions are part of the inferential output.
    \end{itemize}\\
    \hline
  \end{tabular}
  \endgroup
  \caption{%
  Comparing the two types of ranked tree shapes. }
  \label{tab:summary1}
\end{table}

Ranked tree shapes are related to another type of tree structure known by many names including \textit{binary increasing trees}, \textit{ordered binary trees}, or \textit{Andr\'e trees} \cite{foata1971,donaghey1975, poupard1989}.
Ordered (increasing, Andr\'e) binary trees are fully heterochronous ranked tree shapes without the full-binary-tree requirement; such trees have nodes with out-degree at most two instead of out-degree exactly zero or two.
Fully heterochronous ranked tree shapes are also called strictly ordered binary trees.
While isochronous ranked tree shapes are not ordered binary trees, the isochronous ranked tree shapes with $n$ leaves are equinumerous with the ordered binary trees with $n-1$ nodes.
Ordered binary trees are inherently related to alternating permutations that were extensively studied in \cite{andre1881}, which is why some authors call such trees Andr\'e trees.
In the phylogenetics literature, \cite{Gavryushkina2013-eq} developed efficient algorithms for counting fully heterochronous ranked tree shapes using a coalescent-based recursion, and extended this enumeration to serial sampling scenarios where samples are collected at multiple time points.

Let $\mathcal{T}_n$ denote the set of isochronous ranked tree shapes with $n$ leaves and $\mathcal{T}^*_n$ denote the set of fully heterochronous ranked tree shapes with $n$ leaves.
The cardinalities of these sets correspond to the so-called Euler up/down (or zigzag) numbers and reduced tangent numbers \cite{OEIS-iso, OEIS-hetero}.
In particular, in terms of exponential generating functions, we have
\begin{gather*}
  \sum_{n=0}^\infty \frac{|\mathcal{T}_{n+1}|}{n!} x^n
  =
  \sec(x) + \tan(x)
  ,\qquad
  \sum_{n=1}^\infty \frac{|\mathcal{T}^*_{n}| x^{2n}}{(2n)!}
  =
  2\log\left( \sec\left(\frac{x}{\sqrt{2}}\right) \right)
  .
\end{gather*}
Furthermore,
\begin{gather} \label{eq.cardinality}
  |\mathcal{T}_{n}| = 2^{n-1}\left|E_{n-1}\left(\tfrac12\right) - E_{n-1}(0)\right|
  ,\qquad
  |\mathcal{T}^*_{n}| = 2^n\left(2^{2n}-1\right)\frac{|B_{2n}|}{n},
\end{gather}
where $E_n(x)$ are the Euler polynomials (note $E_n(\frac12)=0$ for odd $n$ and $E_n(0)=0$ for even $n$) and $B_n$ are the Bernoulli numbers.

\section{Preliminaries}\label{section.preliminaries}

For a ranked tree shape we define three types of matrices, which we call \textbf{F}-, \textbf{D}-, and \textbf{E}-matrices.
The differences between these matrices and their isochronous analogs are minor, and we highlight where differences occur.
One additional difference with previous work is that we will use the convention that indices start at $0$, not $1$, in order to make theorem statements cleaner.
In formulating the matrices, we give a purely graph theoretic definition and then an interpretation where the total ordering is based on events occurring in time, which is relevant for applications and is useful when visualizing such trees.

Throughout this section, we suppose $T$ is a ranked tree shape with $n$ leaves.
Note $T$ has $n-1$ internal nodes.
We label the nodes of $T$ by their ordering and call this label the \emph{rank} of a node.
As a convention for the isochronous case, we label the leaves with the common rank $n-1$ (distinct ranks are provided for leaves in the fully heterochronous case).
The root has rank 0.

The \textbf{F}-matrix associated to $T$ is a lower triangular matrix $F$, where the size of the matrix is $(n-1)\times(n-1)$ in the isochronous case and $(2n-2)\times(2n-2)$ in the fully heterochronous case.

\begin{defn}\label{def.Fij}
	The entry $F_{i,j}$, for $0\leq j \leq i$, is defined as the number of edges from nodes $v$ to nodes $w$, where the rank of $v$ is at most $j$ and the rank of $w$ is larger than $i$.
\end{defn}

The associated \textbf{D} and \textbf{E}-matrices are also lower triangular matrices and of the same size as the \textbf{F}-matrix.
The entry $D_{i,j}$ is defined as the number of edges descending from the node with rank $j$ to nodes with rank larger than $i$.
The entry $E_{i,j}$ is defined as the number of edges from the node with rank $j$ to the node(s) with rank $i+1$.

For a description of these matrices in line with their introduction in \cite{Kim2020-ip, Samyak2024-yo}, we view $T$ as describing a branching and sampling process of lineages over time.
A coalescent or branching event corresponds to an internal node in $T$, while a sampling event corresponds to a leaf in $T$. 
As in previous literature, we take the approach that time moves in the direction of leaf to root (for instance, one might think of the time units being in millions of years ago).
In the isochronous case, we take real numbers $u_0>u_1>\dotsb>u_{n-1}=0$ and say the event for the node(s) with rank $i$ occurs exactly at time $u_i$.
In the fully heterochronous case we instead take real numbers $u_0>u_1>\dotsb>u_{2n-2}$, as there are more ranked nodes in this case.
No event occurs in any time interval $(u_i,u_{i+1})$.
Exactly one event occurs at each time $u_i$, except for time $u_{n-1}$ in the isochronous case. 

In this setting, entry $(i,j)$ of the \textbf{F}-matrix is the number of lineages present for the entire time interval $(u_{i+1},u_j)$.
A lineage is present for a time interval if the lineage appeared at or before the event time $u_j$ and neither bifurcates nor is sampled before the event time $u_{i+1}$.
Similarly, the $(i,j)$ entry of the \textbf{D}-matrix is the number of direct descendants of the lineage appearing at time $u_j$ that are extant at least until time $u_{i+1}$.
Lastly, the $(i,j)$ entry of the \textbf{E}-matrix is the number of direct descendants of the lineage appearing at time $u_j$ that are sampled at time $u_{i+1}$.

\begin{example}\label{example.heterochronous}
  Consider the following fully heterochronous ranked tree shape on three leaves:
  \begin{figure}[h]
    \includegraphics[width=5cm, height=4cm]{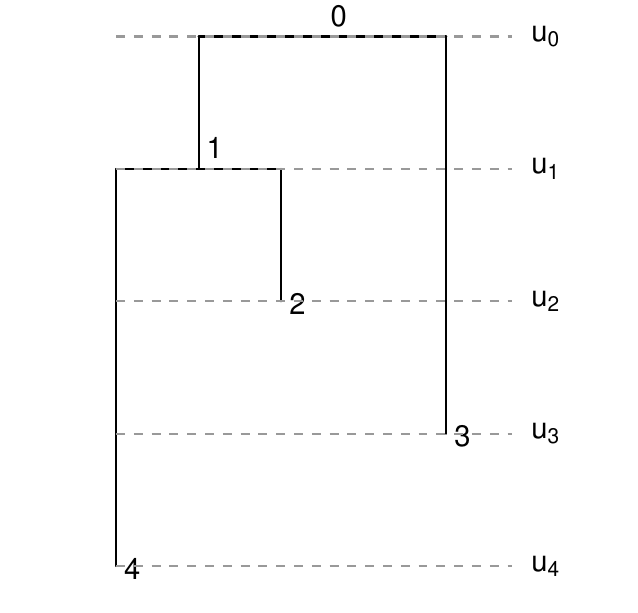}
    \centering
  \end{figure}

  The associated matrices are:
  \begin{gather*}
    F =
    \begin{pmatrix}
      2& 0& 0& 0\\
      1& 3& 0& 0\\
      1& 2& 2& 0\\
      0& 1& 1& 1
    \end{pmatrix}
    ,\quad
    D =
    \begin{pmatrix}
      2& 0& 0& 0\\
      1& 2& 0& 0\\
      1& 1& 0& 0\\
      0& 1& 0& 0
    \end{pmatrix}
    ,\quad
    E=
    \begin{pmatrix}
      1& 0& 0& 0\\
      0& 1& 0& 0\\
      1& 0& 0& 0\\
      0& 1& 0& 0
    \end{pmatrix}
    .
  \end{gather*}

Let us examine the second column of these matrices.
For the \textbf{F}-matrix, we consider the three lineages present just before time $u_1$ (i.e., those below the dashed line labeled $u_1$).
Of these three lineages, all are present by time $u_2$, so $F_{1,1} = 3$; two are present by time $u_3$, so $F_{2,1}=2$; and one is present by time $u_4$, so $F_{3,1}=1$.
For the \textbf{D}-matrix, we consider the two lineages that appear at time $u_1$ (i.e., the lineages with ranks $2$ and $4$).
Of these two lineages, both are present by time $u_2$, so $D_{1,1}=2$; one is present by time $u_3$, so $D_{2,1}=1$; and one is present by time $u_4$, so $D_{3,1}=1$.
For the \textbf{E}-matrix, we again consider the two lineages that appear at time $u_1$.
Of these two lineages, one is sampled at time $u_2$, so $E_{1,1} = 1$; none are sampled at time $u_3$, so $E_{2,1}=0$; and one is sampled at time $u_4$, so $E_{3,1}=1$.
  
\end{example}

\begin{example}\label{example.isochronous}
  Consider the following isochronous ranked tree shape on five leaves:

  \begin{figure}[h]
    \includegraphics[width=5cm, height=4cm]{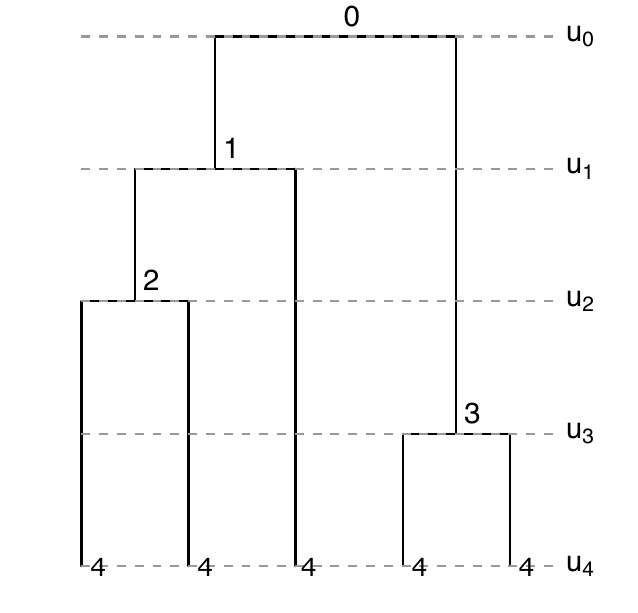}
    \centering
  \end{figure}

  Note that all the leaves have rank 4.
  The associated matrices are:
  \begin{gather*}
    F =
    \begin{pmatrix}
      2& 0& 0& 0\\
      1& 3& 0& 0\\
      1& 2& 4& 0\\
      0& 1& 3& 5
    \end{pmatrix}
    ,\quad
    D =
    \begin{pmatrix}
      2& 0& 0& 0\\
      1& 2& 0& 0\\
      1& 1& 2& 0\\
      0& 1& 2& 2
    \end{pmatrix}
    ,\quad
    E =
    \begin{pmatrix}
      1& 0& 0& 0\\
      0& 1& 0& 0\\
      1& 0& 0& 0\\
      0& 1& 2& 2
    \end{pmatrix}
    .
  \end{gather*}
\end{example}

The entries of such matrices are non-negative integers.
Given that the trees are binary, the entries of a \textbf{D}-matrix are restricted to $\{0,1,2\}$.
In the fully heterochronous case, the entries of an \textbf{E}-matrix are restricted to $\{0,1\}$.
In the isochronous case, the entries of all but the last row of an \textbf{E}-matrix are restricted to $\{0,1\}$, while entries of the last row are restricted to $\{0,1,2\}$ (as the leaves share a common rank).

The \textbf{E}-matrix is related to the adjacency matrix for $T$ as a directed graph.
This is immediate in the fully heterochronous case, where all nodes are uniquely given by their rank, so that $E_{i,j}$, for $0\le j\le i$, is entry $(j,i+1)$ of the adjacency matrix.
In the isochronous case, this is true for all rows of $E$ except the last, where the last row of $E$ is a condensed description of the edges given by the last $n$ columns of the adjacency matrix.
In particular, there is not a unique adjacency matrix for $T$, as the leaves are unlabeled.
By taking any assignment of $n-1,n,\dotsc,2n-2$ as labels for the $n$ leaves we have a different adjacency matrix, but regardless of this choice, $E_{n-2,j}$ is the sum of entries of the adjacency matrix at $(j,n-1), (j,n), \dotsc, (j,2n-2)$.
Due to $T$ being a full binary tree, the information lost going from an adjacency matrix to $E$ is exactly the labeling of leaves.

These types of matrices are related through the equations, for $0\le j\le i$,
\begin{align}\label{eq.matrix_bijections_df_ed}
  D_{i,j} &= F_{i,j} - F_{i,j-1},
  &E_{i,j} &= D_{i,j} - D_{i+1,j},
  \\
  \label{eq.matrix_bijections_fd_de}
  F_{i,j} &= \sum_{\ell=0}^j D_{i,\ell},
  &D_{i,j} &= \sum_{\ell=i}^{2n-3} E_{\ell, j},
\end{align}
with the convention that matrix entries at out of bound indices are 0.
These equations imply a bijection between \textbf{D}-matrices, \textbf{E}-matrices, and \textbf{F}-matrices.

Given the relation between an \textbf{E}-matrix and an adjacency matrix, along with the bijections \eqref{eq.matrix_bijections_df_ed} and \eqref{eq.matrix_bijections_fd_de}, it is clear that \textbf{F}-matrices are in bijection with the ranked tree shapes they represent.
However, it is not apparent how the entries of an \textbf{F}-matrix are constrained or how to tell if a given matrix is an \textbf{F}-matrix.
In the isochronous case, the conditions on entries are known by previous work, which we restate in the following theorem with our notational conventions.

\begin{thm}\cite{Kim2020-ip, Samyak2024-yo}
  \label{theorem.isochronous_f_matrix_tree}
  The space of isochronous ranked tree shapes with $n$ leaves is in bijection with the space of $(n-1) \times (n-1)$ \textbf{F}-matrices, which are lower triangular square matrices of nonnegative integers that obey the following constraints.
  \begin{enumerate}
    \item Entries of rows are monotone increasing:
      \begin{gather*}
        F_{i,j-1}\le F_{i,j} \qquad\qquad\text{for }1\le j\le i\le n-2.
      \end{gather*}
    \item Entries of columns are monotone decreasing with difference at most $1$:
      \begin{gather*}
        F_{i-1,j} - 1 \leq F_{i,j} \leq F_{i-1,j} \qquad\qquad\text{for }0\le j < i \le n-2.
      \end{gather*}
    \item Entries satisfy an additional constraint based on their position in the matrix:
      \begin{enumerate}[label=(\alph*)]
        \item The diagonal elements are $F_{i,i}=i+2$.
        \item The subdiagonal elements are $F_{i,i-1} = i$ for $1\le i\le n-2$.
        \item Of the remaining elements, $F_{i,j}$  for $2\le i \le n-2$ and $1\le j \le i-2$, satisfy the inequality
          \begin{gather*}
            F_{i,j-1} + F_{i-1,j} - F_{i-1,j-1} - 1 \leq F_{i,j} \leq F_{i,j-1} + F_{i-1,j} - F_{i-1,j-1}.
          \end{gather*}
      \end{enumerate}
  \end{enumerate}
\end{thm}

A consequence of Theorem~\ref{theorem.isochronous_f_matrix_tree} is that it allows us to enumerate the whole space of isochronous ranked tree shapes, with a fixed number of leaves, via $\mathbf{F}$-matrices.
The values of the diagonal and subdiagonal entries are common to all $\mathbf{F}$-matrices.
The $\mathbf{F}$-matrices are enumerated by then selecting values for the remaining lower diagonal entries in lexicographical order, which is the order of rows then columns (see Example \ref{ex.isochronous_fill}).
As it turns out, selecting values in this order not only produces all \textbf{F}-matrices, but also never produces an invalid matrix.
That is to say, setting $F_{i,j}$ to either $\min\left(F_{i-1,j}, F_{i,j-1}+F_{i-1,j}-F_{i-1,j-1}\right)$ or
$\max\left(F_{i,j-1}, F_{i-1,j}-1, F_{i,j-1}+F_{i-1,j}-F_{i-1,j-1}-1\right)$
never yields an unsatisfiable system of inequalities for entries filled after $F_{i,j}$.

The simplicity of this result motivates the use of \textbf{F}-matrices over \textbf{D}- or \textbf{E}-matrices.
We state the corresponding theorem for fully heterochronous ranked tree shapes in the next section, however in this case, the enumeration method is not immediate.

\section{Theorems}\label{section.theorems}

We note that the $\mathbf{F}$-matrix of a fully heterochronous ranked tree with $n$ leaves is a matrix of dimension $(2n-2) \times (2n-2)$ with a different constraint on the diagonal from the isochronous case.
Recall the $i$-th diagonal entry indicates the number of lineages (or edges) extant at the $i$-th time epoch and so the diagonal entries either increase by one or decrease by one depending on whether the $i$-th node is of out-degree 2 or of out-degree 0.
In the isochronous case, diagonal entries always increase by one.
In the following theorem, we classify $\textbf{F}$-matrices of fully heterochronous ranked tree shapes in terms of a system of inequalities.

\begin{thm}\label{theorem.heterochronous_f_matrix_tree}
  The space of fully heterochronous ranked tree shapes with $n$ leaves is in bijection with the space of $(2n-2) \times (2n-2)$ \textbf{F}-matrices, which are the lower triangular square matrices $F$ of non-negative integers that obey the following constraints.
  \begin{enumerate}
    \item\label{item.row} Entries of rows are monotone increasing:
      \begin{gather*}
        F_{i,j-1}\le F_{i,j} \qquad\qquad\text{for }1\le j\le i\le 2n-3.
      \end{gather*}
    \item\label{item.column} Entries of columns are monotone decreasing with difference at most $1$:
      \begin{gather*}
        F_{i-1,j} - 1 \leq F_{i,j} \leq F_{i-1,j} \qquad\qquad\text{for }0\le j < i \le 2n-3.
      \end{gather*}
    \item\label{item.positional} Entries satisfy an additional constraint based on their position in the matrix:
      \begin{enumerate}[label=(\alph*)]
        \item\label{item.diagonal} The diagonal elements are positive and satisfy,
          \begin{align*}
            F_{0,0} &= 2,
            \\
            F_{i,i} &= F_{i-1,i-1} \pm 1 \qquad\qquad\text{for } 0<i<2n-3,
            \\
            F_{2n-3,2n-3} &= 1.
          \end{align*}
          In particular, $F_{i,i}=F_{i-1,i-1}-1$ if the $i$-th event is a sampling event, and $F_{i,i}=F_{i-1,i-1}+1$ if it is a coalescent event.
        \item\label{item.subdiagonal} The subdiagonal elements are $F_{i,i-1} = F_{i-1,i-1}-1$ for $1\le i\le 2n-3$.
        \item\label{item.remaining} Of the remaining elements, $F_{i,j}$ for $2\le i \le 2n-3$ and $1\le j \le i-2$, satisfy the inequality
          \begin{gather*}
            F_{i,j-1} + F_{i-1,j} - F_{i-1,j-1} - 1 \leq F_{i,j} \leq F_{i,j-1} + F_{i-1,j} - F_{i-1,j-1}.
          \end{gather*}
      \end{enumerate}
  \end{enumerate}
\end{thm}
\begin{proof}

  We first verify that the conditions are necessary.
  Suppose $F$ is the \textbf{F}-matrix associated to a fully heterochronous ranked tree shape with $n$ leaves.
  Let $D$ and $E$ be the associated \textbf{D}-matrix and \textbf{E}-matrix.

  Since $F_{i,j} - F_{i,j-1} = D_{i,j}$, Condition \ref{item.row} is equivalent to $D_{i,j}\ge 0$, which is true.
  Condition \ref{item.column} states that the number of edges from nodes $v$ to nodes $w$, where $\text{rank}(v)\le j$ and $\text{rank}(w)=i$, is exactly 1 or 0 (either the parent node of $w$ has rank at most $j$ or not).

  We handle each part of Condition \ref{item.positional} in order of appearance.
  Since the root node has exactly two children, $F_{0,0}=2$.
  The same edges are counted by $F_{i-1,i-1}$ and $F_{i,i}$ except for three: the edge to the node with rank $i$ (counted by $F_{i-1,i-1}$) and the two edges from the node with rank $i$ (counted by $F_{i,i}$, if they exist), so $F_{i,i}-F_{i-1,i-1}=\pm1$.
  There is a single node with rank larger than $2n-3$, so $F_{2n-3,2n-3}=1$.
  The same edges are counted by $F_{i-1,i-1}$ and $F_{i,i-1}$ except the edge to the node of rank $i$ (counted by $F_{i-1,i-1}$), so $F_{i-1,i-1}-F_{i,i-1}=1$.
  Since $F_{i,j-1} - F_{i,j} + F_{i-1,j}-F_{i-1,j-1} = D_{i-1,j}- D_{i,j} = E_{i-1,j}$, Condition \ref{item.positional}\ref{item.remaining} is equivalent to $E_{i-1,j}\in\{0,1\}$, which is true.

  Next we prove that the conditions are sufficient.
  It is easier to work with the $\mathbf{D}$- and $\mathbf{E}$-matrices than directly with the $\mathbf{F}$-matrix.
  Suppose $F$ is a matrix satisfying the conditions in the statement of the theorem.
  Let $D$ and $E$ be the matrices defined by \eqref{eq.matrix_bijections_df_ed}.
  We show that $E$ is the offset adjacency matrix of some totally ranked tree shape, i.e., $E_{i,j} = A_{j,i+1}$ where $A$ is the adjacency matrix. 
  Recall that in the adjacency matrix, entry $A_{j,i+1}=1$ if nodes with ranks $j$ and $i+1$ are connected by an edge. 
  This requires verifying the following conditions for $E$:
  \begin{enumerate}[label=(\roman*)]
    \item\label{item.E.values} Each $E_{i,j}\in\{0,1\}$, as these are the only valid entries of an adjacency matrix.
    \item\label{item.E.rows} Each row sums to $1$, $\sum_{j=0}^i E_{i,j}=1$, as no node has multiple parents and there is exactly one event (coalescent or sampling) at each event time.
    \item\label{item.E.cols} Each column sums to $0$ or $2$, $\sum_{i=j}^{2n-3} E_{i,j}\in\{0,2\}$, as the tree is binary.
  \end{enumerate}
  This will complete the proof, as we can read the ranked tree shape from the matrix $E$.

  By the definitions of the matrices $E$ and $D$, along with Condition \ref{item.positional}\ref{item.remaining}, we have
  \begin{align*}
    E_{i,j}
    &=
    D_{i,j} - D_{i+1,j}
    =
    F_{i,j} - F_{i,j-1} - F_{i+1,j} + F_{i+1,j-1}
    \\
    &=
    -\left( F_{i+1,j} - F_{i+1,j-1} - F_{i,j} + F_{i,j-1} \right)
    =
    0 \text{ or } 1
    ,
  \end{align*}
  which is \ref{item.E.values}.
  With \ref{item.positional}\ref{item.subdiagonal}, or \ref{item.positional}\ref{item.diagonal} when $i=2n-3$, we have
  \begin{align*}
    \sum_{j=0}^i E_{i,j}
    &=
    \sum_{j=0}^i (D_{i,j} - D_{i+1,j})
    =
    \sum_{j=0}^i (F_{i,j} - F_{i,j-1} - F_{i+1,j} + F_{i+1,j-1})
    \\
    &=
    F_{i,i} - F_{i+1, i}
    = 1
    ,
  \end{align*}
  which is \ref{item.E.rows}.
  By \ref{item.positional}\ref{item.subdiagonal} and \ref{item.positional}\ref{item.diagonal}, we have
  \begin{align*}
    \sum_{i=j}^{2n-3} E_{i,j}
    &=
    \sum_{i=j}^{2n-3} (D_{i,j} - D_{i+1,j})
    =
    D_{j,j}
    =
    F_{j,j} - F_{j,j-1}
    \\
    &=
    \begin{cases}
      F_{0,0} & \text{if } j=0
      ,\\
      F_{j,j} - F_{j-1,j-1} + 1 & \text{otherwise},
    \end{cases}
    \\
    &=
    \begin{cases}
      2 & \text{if } j=0
      ,\\
      0 \text{ or } 2 & \text{otherwise},
    \end{cases}
  \end{align*}
  which is \ref{item.E.cols}.

\end{proof}

With Theorems \ref{theorem.isochronous_f_matrix_tree} and \ref{theorem.heterochronous_f_matrix_tree}, we can tell if a given matrix represents a ranked tree shape or not.
While the difference between the two cases is the diagonal, this is more important than it appears.

We next emphasize the difference between the two cases by showing how a matrix-filling strategy that works for the isochronous case will produce invalid \textbf{F}-matrices in the heterochronous case.
We will then develop a strategy (Proposition~\ref{prop.extend_sequences}) that can fill the matrix in a single pass.

\begin{example}\label{ex.isochronous_fill}
  The \textbf{F}-matrices for isochronous ranked tree shapes with five leaves must fit the pattern:
  \begin{gather*}
    F =
    \begin{pmatrix}
      2 & 0 & 0 & 0\\
      1 & 3 & 0 & 0\\
      * & 2 & 4 & 0\\
      * & * & 3 & 5
    \end{pmatrix}.
  \end{gather*}
  We can determine all \textbf{F}-matrices by filling the remaining entries in order of $F_{2,0}$, $F_{3,0}$, and $F_{3,1}$.
  For $F_{2,0}$ we have two options: 0 or 1.
  Suppose we select $F_{2,0}=0$.
  Moving to $F_{3,0}$, we are forced to select $F_{3,0}=0$ by Constraint 2.
  Lastly, for $F_{3,1}$ our options are 1 or 2, both of which yield valid \textbf{F}-matrices.
  One can verify that if we instead begin with $F_{2,0}=1$, the remaining entries work out in a similar fashion.

  To see what can go wrong in the fully heterochronous case, consider the partially filled \textbf{F}-matrix,
  \begin{gather*}
    F =
    \begin{pmatrix}
      2 & 0 & 0 & 0\\
      1 & 3 & 0 & 0 \\
      * & * & * & 0\\
      * & * & * & *
    \end{pmatrix}.
  \end{gather*}
  For $F_{2,0}$ we have two options, 0 or 1; suppose we take $F_{2,0}=0$.
  We are forced to have $F_{2,1}=2$ by Condition \ref{item.positional}\ref{item.subdiagonal}.
  Next we must take $F_{2,2}=2$, as $F_{2,2}=4$ yields the contradiction $3=F_{3,2}\leq F_{3,3}=1$ by Conditions \ref{item.positional}(a,b).
  Additionally, we are forced to take $F_{3,0}=0$ by Constraint 2.
  In
  \begin{gather*}
    F =
    \begin{pmatrix}
      2 & 0 & 0 & 0\\
      1 & 3 & 0 & 0 \\
      0 & 2 & 2 & 0\\
      0 & * & * & *
    \end{pmatrix},
  \end{gather*}
  we have the two options of 1 or 2 for $F_{3,1}$, but are forced to have $F_{3,2}=F_{3,3}=1$ by \ref{item.positional}(a,b).
  While
  \begin{gather*}
    F =
    \begin{pmatrix}
      2 & 0 & 0 & 0\\
      1 & 3 & 0 & 0 \\
      0 & 2 & 2 & 0\\
      0 & 1 & 1 & 1
    \end{pmatrix}
  \end{gather*}
  is a valid \textbf{F}-matrix,
  \begin{gather*}
    F =
    \begin{pmatrix}
      2 & 0 & 0 & 0\\
      1 & 3 & 0 & 0 \\
      0 & 2 & 2 & 0\\
      0 & 2 & 1 & 1
    \end{pmatrix}
  \end{gather*}
  is not as the last row violates the monotone increasing property.
\end{example}

This example shows how the strategy of filling rows one by one from top to bottom (i.e., traversing entries $F_{0,0}, F_{1,0}, F_{1,1}, F_{2,0}, \cdots, F_{2n-3,2n-4}, F_{2n-3,2n-3}$)  produces all \textbf{F}-matrices, in both the isochronous and heterochronous case, but additional constraints are necessary to prevent invalid \textbf{F}-matrices in the heterochronous case.
Specifically, some combinations of values for $F_{i,j}$ and $F_{i,i-1}$ from Conditions \ref{item.positional}\ref{item.remaining} and \ref{item.positional}\ref{item.subdiagonal} in Theorem~\ref{theorem.heterochronous_f_matrix_tree} may conflict with item \ref{item.row}.
In the example, the invalid combination is $F_{i,j}=F_{3,1}=2$ and $F_{i,i-1}=F_{3,2}=1$.

For the remainder of this section, we describe a matrix-filling strategy that does not lead to contradictions in the heterochronous case.
As the subdiagonal entries are determined by the diagonal entries, we first verify that any choice of diagonal entries by Condition \ref{item.positional}\ref{item.diagonal} and an additional constraint yields at least one valid $\mathbf{F}$-matrix.
That is to say, when solving the system of inequalities in Theorem \ref{theorem.heterochronous_f_matrix_tree}, we may select values for the diagonal without backtracking.

\begin{cor}\label{cor.totally_ranked_diagonals}
  Let $n$ and $N$ be non-negative integers with $N\le 2n-3$.
  Suppose $f_{i}$, for $0\le i\le N$, is a sequence of positive integers where,
  \begin{enumerate}
    \item $f_{0}=2$,
    \item\label{item.cor.totally_ranked_diagonals.2} $f_{i} = f_{i-1}\pm 1$ for $1\le i\le N$, and
    \item\label{item.cor.totally_ranked_diagonals.3} $f_{i}\leq 2n-i-2$ for $0\le i\le N$.
  \end{enumerate}
  Then there exists $F$, an \textbf{F}-matrix for a fully heterochronous ranked tree shape with $n$ leaves, with $F_{i,i}=f_{i}$ for $0\le i\le N$.
\end{cor}
\begin{proof}
  The inequality in Item \ref{item.cor.totally_ranked_diagonals.3} of the corollary guarantees that it is possible to extend the sequence to length $2n-3$ while satisfying item \ref{item.cor.totally_ranked_diagonals.2}, Item \ref{item.cor.totally_ranked_diagonals.3}, and $f_{2n-3}=1$.
  The $f_{i}$ are the first $N+1$ diagonal entries of any $(2n-2)\times(2n-2)$ \textbf{F}-matrix associated to a fully heterochronous ranked tree shape whose first $N+1$ nodes (ordered by rank) bifurcate when $f_{i} = f_{i-1}+1$ and are leaves when $f_{i} = f_{i-1}-1$.
\end{proof}

We introduce notation for bounds that often appear with \textbf{F}-matrices.
For a matrix or doubly indexed sequence, $F$, we set
\begin{align*}
  L_F(i,j)  &:= \max\left(F_{i,j-1}, F_{i-1,j}-1, F_{i,j-1} + F_{i-1,j} - F_{i-1,j-1} - 1\right)
  ,\\
  U_F(i,j) &:= \min(F_{i-1,j}, F_{i,j-1} + F_{i-1,j} - F_{i-1,j-1}),
\end{align*}
with the convention that $F_{k,\ell}=0$ when $k$ or $\ell$ is negative.
The key feature of these bounds is that the entries of an \textbf{F}-matrix, off the diagonal and subdiagonal, are classified by $L_F(i,j) \le F_{i,j}\le U_F(i,j)$.
When filling the entries of an \textbf{F}-matrix, by row then column, we are free to choose $L_F(i,j)$ or $U_F(i,j)$ for $F_{i,j}$ in the isochronous case, but this is not always true in the heterochronous case.

We require notation for the concept of a partially filled \textbf{F}-matrix of a fully heterochronous ranked tree shape.
This will correspond to filling the first $N$ rows and the first $M+1$ columns of the $N+1$st row of an \textbf{F}-matrix.
We must do so in a way that guarantees the values chosen so far will not conflict with values chosen later on.
\begin{defn}\label{def.f_sequence}
  Let $n$, $N$, and $M$ be non-negative integers with $M\le N\le 2n-3$ and set $B=\max(N-1, M)$.
  An \emph{$(n,N,M)$ \textbf{F}-sequence} is a doubly indexed sequence $f_{i,j}$, defined for $(i,j)$ in $\left\{(i,j)\mid 0\le j\le i\le N-1\right\} \cup \left\{(N,j)\mid 0\le j\le M\right\}$, of non-negative integers where,
  \begin{enumerate}
    \item\label{def.partial_seq.diagonal} the sequence $f_{i,i}$, for $0\le i \le B$, satisfies the conditions of Corollary \ref{cor.totally_ranked_diagonals} with $(n,N)\mapsto(n,B)$,
    \item\label{def.partial_seq.general_bounds} $L_f(i,j) \le f_{i,j} \le U_f(i,j)$ for $0\le j\le i-2$, where $(i,j)$ are valid indices, and
    \item\label{def.partial_seq.special} if $f_{i-1,j}=f_{i-1,i-1}$, then $f_{i,j} = f_{i-1,i-1} - 1$, for valid indices with $0\le j\le i-1$.
  \end{enumerate}
\end{defn}
We will show below that an $(n,N,M)$ \textbf{F}-sequence fills the first $N$ rows and the first $M+1$ columns of the $N+1$st row of a $(2n-2)\times(2n-2)$ \textbf{F}-matrix.
Under the lexicographical order, the \textbf{F}-matrix is filled up to and including entry $(N,M)$.

\begin{example}\label{ex.f_sequence}
  As seen in Example \ref{ex.isochronous_fill}, there is one $(5,1,1)$ \textbf{F}-sequence: $
  \begin{psmallmatrix}2&\\1&3
  \end{psmallmatrix}$.
  The two possible $(5,2,0)$ \textbf{F}-sequences extend the $(5,1,1)$ \textbf{F}-sequence by filling the first entry of the next row, and
  are given by $
  \begin{psmallmatrix}  2&\\1&3\\0&
  \end{psmallmatrix}$   and $
  \begin{psmallmatrix}2&\\1&3\\  1&
  \end{psmallmatrix}$.
\end{example}

The following lemma is necessary for extending an \textbf{F}-sequence by one entry.
\begin{lemma}
\label{lemma.prop_extend.inequality}
If $f_{i,j}$ is an $(n,N,M)$ \textbf{F}-sequence with $M<N-2$, then 			
\begin{gather*}
	f_{N,M} \le f_{N-1,M} \le f_{N-1,M+1}.
\end{gather*}
\end{lemma}
\begin{proof}
	By Definition \ref{def.f_sequence}.\ref{def.partial_seq.general_bounds}, $f_{N,M} \le U_f(N,M) \le f_{N-1,M}$.
	We handle the remaining bound in three cases.
	First note that by Definitions \ref{def.f_sequence}.\ref{def.partial_seq.special} and \ref{def.f_sequence}.\ref{def.partial_seq.diagonal},
	\begin{gather*}
		f_{N-2,N-3} = f_{N-3,N-3} - 1 = (f_{N-2,N-2} \pm 1) -1.
	\end{gather*}
	When $M=N-3$ and $f_{N-2,N-3} = f_{N-2,N-2}$, by Definition \ref{def.f_sequence}.\ref{def.partial_seq.special},
	\begin{gather*}
		f_{N-1,M} = f_{N-1,N-3} = f_{N-2,N-2}-1 = f_{N-1,N-2} = f_{N-1,M+1}.
	\end{gather*}
	When $M=N-3$ and $f_{N-2,N-3} = f_{N-2,N-2}-2$, Definitions \ref{def.f_sequence}.\ref{def.partial_seq.general_bounds} and \ref{def.f_sequence}.\ref{def.partial_seq.special} give
	\begin{align*}
		f_{N-1,M} \le U_f(N-1,N-3) \le f_{N-2,N-3} = f_{N-2,N-2} - 2
		&= f_{N-1,N-2}-1
		\\
		&< f_{N-1,M+1}.
	\end{align*}
	Lastly, when $M<N-3$, by Definition \ref{def.f_sequence}.\ref{def.partial_seq.general_bounds},
	\begin{gather*}
		f_{N-1,M+1} \ge L_f(N-1,M+1) \ge f_{N-1,M}.
	\end{gather*}
\end{proof}

We show how to extend an $(n,N,M)$ \textbf{F}-sequence to a $(2n-2) \times (2n-2)$ \textbf{F}-matrix and that every \textbf{F}-matrix appears in such a way.

\begin{prop}\label{prop.extend_sequences}
  If $f_{i,j}$ is an $(n,N,M)$ \textbf{F}-sequence with $M<N$ or $N<2n-3$, then the following methods extend $f$ to a longer \textbf{F}-sequence.
  \begin{enumerate}
    \item\label{item.extend_sequences.row} If $M<N-2$, setting $f_{N,M+1}$ to
      \begin{enumerate}
        \item $f_{N-1,N-1}-1$ if $f_{N-1,M+1}=f_{N-1,N-1}$, and otherwise
        \item either of $L_f(N,M+1)$ or $U_f(N,M+1)$,
      \end{enumerate}
      yield $(n,N,M+1)$ \textbf{F}-sequences.
    \item\label{item.extend_sequences.subdiagonal} If $M=N-2$, setting $f_{N,N-1}=f_{N-1,N-1}-1$ yields an $(n,N,N-1)$ \textbf{F}-sequence.
    \item\label{item.extend_sequences.diagonal} If $M=N-1$, setting $f_{N,N}$ to either
      \begin{itemize}
        \item $f_{N-1,N-1}-1$ if $f_{N-1,N-1}>1$, or
        \item $f_{N-1,N-1}+1$ if $f_{N-1,N-1}<2n-N-1$,
      \end{itemize}
      both yield $(n,N,N)$ \textbf{F}-sequences.
    \item\label{item.extend_sequences.new_row}
      If $M=N$, setting $F_{N+1,0}$ to
      \begin{enumerate}
        \item $f_{N,N}-1$ if $f_{N,0}=f_{N,N}$, and otherwise
        \item either of $\max\left(0, f_{N,0}-1\right)$ or $f_{N,0}$,
      \end{enumerate}
      yield $(n,N+1,0)$ \textbf{F}-sequences.
  \end{enumerate}
\end{prop}
\begin{proof}
	The proof is to verify that the extended sequences satisfy Definition \ref{def.f_sequence}.
	We do so in the order of the four cases of the size of $M$ relative to $N$.
	However, we do not split each case into the subcases of the statement of the proposition.
    
	We begin with Item \ref{item.extend_sequences.row}, where we assume $M<N-2$.
	This is the lengthiest case to prove.	
	Since the sequence is not extended with a diagonal entry, the extended sequence immediately satisfies Definition \ref{def.f_sequence}.\ref{def.partial_seq.diagonal}.
	To verify the extended sequence satisfies Definitions \ref{def.f_sequence}.\ref{def.partial_seq.general_bounds} and \ref{def.f_sequence}.\ref{def.partial_seq.special}, we show that $L_f(N,M+1) \leq U_f(N,M+1)$ and that if $f_{N-1,M+1}=f_{N-1,N-1}$ then $L_f(N,M+1)=f_{N-1,N-1}-1$.
		
	At the outset, we do not know that $L_f(N,M+1) \leq U_f(N,M+1)$; we only have that $L_f(i,j)\leq f_{i,j} \leq U_f(i,j)$ for $0\le j\le i-2$ with $(i,j)$ before ${(N,M+1)}$ in lexicographical order.
	Recall that
	\begin{gather*}
	U_f(N,M+1) = \min(f_{N-1,M+1}, f_{N,M} + f_{N-1,M+1} - f_{N-1,M}).
	\end{gather*}
	We consider the two possible values for $U_f(N,M+1)$.
  	
	Suppose $U_{f}(N,M+1)=f_{N-1,M+1}$, which means $f_{N,M}\ge f_{N-1,M}$.
	However, $f_{N,M}\le f_{N-1,M} \le f_{N-1,M+1}$ by Lemma \ref{lemma.prop_extend.inequality}, so in fact $f_{N,M} = f_{N-1,M}$,
	and therefore,
  	\begin{gather*}
    	L_f(N,M+1)
    	=
    	\max\left(f_{N,M}, f_{N-1,M+1} - 1\right)
    	\le
    	f_{N-1,M+1}
   	 =
    	U_f(N,M+1).
	\end{gather*}
	What remains is to show that if $f_{N-1,M+1}=f_{N-1,N-1}$ then $L_f(N,M+1)=f_{N-1,N-1}-1$.
	Suppose $f_{N-1,M+1}=f_{N-1,N-1}$.
	We claim that $f_{N,M}<f_{N-1,M+1}$.
	For if not, meaning $f_{N,M}=f_{N-1,M+1}=f_{N-1,N-1}$, then $f_{N-1,M}=f_{N-1,M+1}=f_{N-1,N-1}$ by Lemma \ref{lemma.prop_extend.inequality} and so $f_{N,M}=f_{N-1,N-1}-1$ by Definition \ref{def.f_sequence}.\ref{def.partial_seq.special}, which is a contradiction.
	So if $f_{N-1,M+1}=f_{N-1,N-1}$ then $L_f(N,M+1)=f_{N-1,N-1}-1$.
	In summary, under the assumption that $U_{f}(N,M+1)=f_{N-1,M+1}$, the extended sequence satisfies Definition \ref{def.f_sequence}.

	Suppose instead $U_{f}(N,M+1)=f_{N,M}+f_{N-1,M+1}-f_{N-1,M}<f_{N-1,M+1}$, which means $f_{N,M}<f_{N-1,M}$.
	However, $f_{N-1,M}-1\leq f_{N,M}$ by Definition \ref{def.f_sequence}.\ref{def.partial_seq.general_bounds}, so in fact $f_{N,M}=f_{N-1,M}-1$, whereas $f_{N-1,M}\le f_{N-1,M+1}$ by Lemma \ref{lemma.prop_extend.inequality}.
	Therefore
  \begin{align*}
    L_f(N,M+1)
    &=
    \max\left(f_{N,M}, f_{N-1,M+1} - 1\right)
    =
    f_{N-1,M+1} - 1
    \\
    &=
    f_{N-1,M+1} + f_{N,M} - f_{N-1,M}
    =U_f(N,M+1).
  \end{align*}
	If $f_{N-1,M+1}=f_{N-1,N-1}$, then trivially $L_f(N,M+1)=f_{N-1,N-1}-1$.
	In summary, under the assumption that $U_{f}(N,M+1)\not=f_{N-1,M+1}$, the extended sequence satisfies Definition \ref{def.f_sequence}.
	This completes the proof of Item \ref{item.extend_sequences.row}.

	For Item \ref{item.extend_sequences.subdiagonal}, we note that $f_{N-1,N-1}-1$ is exactly the value of $f_{N,N-1}$ specified by Definition \ref{def.f_sequence}.\ref{def.partial_seq.special}	.
	Since the sequence is extended by a subdiagonal entry, the extended sequence immediately satisfies Definitions \ref{def.f_sequence}.\ref{def.partial_seq.diagonal} and \ref{def.f_sequence}.\ref{def.partial_seq.general_bounds}. 
	
	For Item \ref{item.extend_sequences.diagonal}, we note that these values of $f_{N,N}$ are exactly those specified by Definition \ref{def.f_sequence}.\ref{def.partial_seq.diagonal}.
	Since the sequence is extended by a diagonal entry, the extended sequence immediately satisfies Definitions \ref{def.f_sequence}.\ref{def.partial_seq.general_bounds} and \ref{def.f_sequence}.\ref{def.partial_seq.special}. 

	Last, we verify Item \ref{item.extend_sequences.new_row}.
	Since the sequence is not extended with a diagonal entry, the extended sequence immediately satisfies Definition \ref{def.f_sequence}.\ref{def.partial_seq.diagonal}.
	As with Item \ref{item.extend_sequences.row}, we show that the extended sequence satisfies Definitions \ref{def.f_sequence}.\ref{def.partial_seq.general_bounds} and \ref{def.f_sequence}.\ref{def.partial_seq.special} by showing that $L_f(N+1,0) \leq U_f(N+1,0)$ and that if $f_{N,0}=f_{N,N}$ then $L_f(N+1,0)=f_{N,N}-1$.
	We have
  	\begin{gather*}
   	 L_f(N+1,0) = \max(0,f_{N,0}-1) \le f_{N,0} = U_f(N+1,0).
  \end{gather*}
  If $f_{N,0}=f_{N,N}$, then $f_{N,0}\ge 1$ and so $L_f(N+1,0)=f_{N,0}-1=f_{N-1,N-1}-1$.

\end{proof}

\begin{example}
	Using Proposition \ref{prop.extend_sequences} above, we can avoid the conflict in Example \ref{ex.isochronous_fill}. 
	To extend the  $(3,3,0)$ \textbf{F}-sequence 
	\[
	\begin{pmatrix}
		2 & 0 & 0 & 0\\
		1 & 3 & 0 & 0 \\
		0 & 2 & 2 & 0\\
		0 & * & * & *
	\end{pmatrix},
	\]
	Item \ref{item.extend_sequences.row}(a) tells us $f_{3,1}$ must be $1$, and the resulting $(3,3,1)$ \textbf{F}-sequence is 
	\[
	\begin{pmatrix}
		2 & 0 & 0 & 0\\
		1 & 3 & 0 & 0 \\
		0 & 2 & 2 & 0\\
		0 &  1 & * & *
	\end{pmatrix}. 
	\]
	Notably, the illegal value $f_{3,2}=2$ is no longer allowed. 
\end{example}

\begin{cor}\label{corollary.valid_values}
  A $(2n-2)\times(2n-2)$ lower triangular matrix $F$ is an \textbf{F}-matrix for a fully heterochronous ranked tree shape if and only if the entries $F_{i,j}$ are an $(n,2n-3,2n-3)$ \textbf{F}-sequence.
\end{cor}
\begin{proof}
  Suppose $F$ is an \textbf{F}-matrix.
  The sequence $F_{i,j}$ immediately satisfies all conditions of an $(n,2n-3,2n-3)$ \textbf{F}-sequence, except possibly the condition when $F_{i-1,j}=F_{i-1,i-1}$ in Definition \ref{def.f_sequence}.\ref{def.partial_seq.special}.
  If $F_{i-1,j}=F_{i-1,i-1}$ and $0\le j\le i-1$, then by Conditions \ref{item.column}, \ref{item.row}, and \ref{item.positional}\ref{item.subdiagonal} of Theorem \ref{theorem.heterochronous_f_matrix_tree},
  \begin{gather*}
    F_{i-1,j} -1 \le F_{i,j} \le F_{i,i-1} = F_{i-1,i-1} -1.
  \end{gather*}
  So
  \begin{gather*}
    F_{i-1,i-1} -1 \le F_{i,j} \le F_{i-1,i-1} -1,
  \end{gather*}
  meaning $F_{i,j} = F_{i-1,i-1} -1$.

  For the converse, suppose $F_{i,j}$ is an $(n,2n-3,2n-3)$ \textbf{F}-sequence.
  Given the definition of such a sequence and Theorem \ref{theorem.heterochronous_f_matrix_tree}, we need only show that
  $F_{i,i-2} \le F_{i,i-1}$ for $i\ge2$.
  If $F_{i-1,i-2}=F_{i-1,i-1}$, then
  \begin{gather*}
    F_{i,i-2} = F_{i-1,i-1} -1 = F_{i,i-1}.
  \end{gather*}
  Otherwise, $F_{i-1,i-2}=F_{i-1,i-1}-2$ and so
  \begin{gather*}
    F_{i,i-2} \le F_{i-1,i-2} < F_{i-1,i-1} -1 = F_{i,i-1}.
  \end{gather*}
\end{proof}

Let us emphasize that Proposition \ref{prop.extend_sequences} provides the rules used to construct the entries, one at a time, of an \textbf{F}-matrix.
When following these rules, there is no chance of entering an invalid state that requires backtracking to previously selected entries.
Furthermore, when constructing an entry, we need only consider values at four previous entries (the entries directly to the left, directly above, and directly to the above-left, as well as the previous diagonal entry).
As such we have an efficient process to determine all \textbf{F}-matrices of a given size and so all ranked tree shapes on a given number of leaves.
Note it is a straightforward process to turn an \textbf{F}-matrix into an \textbf{E}-matrix using the equations in \eqref{eq.matrix_bijections_df_ed}, and to turn an \textbf{E}-matrix into a ranked tree shape.
This method of enumeration is illustrated for the $4\times4$ \textbf{F}-matrices in Figure \ref{fig:ex_enum1}, with the corresponding fully heterochronous ranked tree shapes drawn in Figure \ref{fig:ex_enum2}.

\begin{figure}[h!]
	\centering
	\begin{forest}
		for tree={
			draw,
			rounded corners,
			align=center,
			inner sep=2pt,
			s sep=12mm, 
			l sep=7mm  
		}
		[
		{$\begin{psmallmatrix}
				2 &  \\
				1 & *
			\end{psmallmatrix}$}
		[
		{$F^0 = \begin{psmallmatrix}
				2 & & &   \\
				1 & 1 & & \\
				0 & 0 & 2 & \\
				0 & 0 & 1 & 1 
			\end{psmallmatrix}$}
		]
		[
		{$\begin{psmallmatrix}
				2 & &  \\
				1 & 3 & \\
				* & & 
			\end{psmallmatrix}$}
		[{$F^1 = \begin{psmallmatrix}
				2 & & &   \\
				1 & 3 & & \\
				0 & 2 & 2 & \\
				0 & 1 & 1 & 1 
			\end{psmallmatrix}$}]
		[{$\begin{psmallmatrix}
				2 & & &   \\
				1 & 3 & & \\
				1 & 2 & 2 & \\
				* &  &   &   
			\end{psmallmatrix}$} 
		[{$F^2 = \begin{psmallmatrix}
				2 & & &   \\
				1 & 3 & & \\
				1 & 2 & 2 & \\
				0 & 1 & 1 & 1 
			\end{psmallmatrix}$}]
		[{$F^3 = \begin{psmallmatrix}
				2 & & &   \\
				1 & 3 & & \\
				1 & 2 & 2 & \\
				1 & 1 & 1 & 1 
			\end{psmallmatrix}$}]]
		]
		]
	\end{forest}
	\caption{The iterative construction of all $4 \times 4$ \textbf{F}-matrices. 
	The positions that admit two possible values are indicated by the $*$ sign, with the two downward edges representing each of the values.
	Intermediate steps for positions with a single possible value are omitted.
	The leaves ($F^0, F^1, F^2, F^3$) enumerate all $4 \times 4$ \textbf{F}-matrices.}
	\label{fig:ex_enum1}
\end{figure}
	
\begin{figure}[h!]
	\centering
	\includegraphics[width=0.8\textwidth]{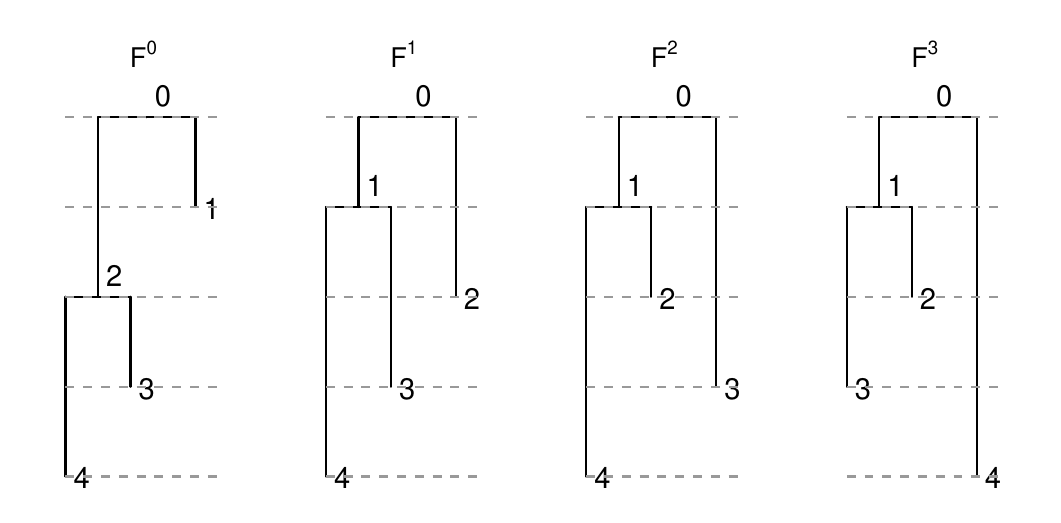}
	\caption{The $3$-leaf fully heterochronous ranked tree shapes corresponding to each of the $4 \times 4$ \textbf{F}-matrices.}
	\label{fig:ex_enum2}
\end{figure}

\section{Sampling Schemes} \label{section.enumeration}

A fully heterochronous ranked tree shape with $n$ leaves can be converted into an isochronous ranked tree shape with $2n$ leaves by attaching isochronous cherries to each of the leaves.
In this case, we will call such a tree a \textit{full-cherry tree}.
A \textit{cherry} is a pair of sister leaves, 
i.e.\ a subgraph with 3 nodes in which the root node has out-degree 2 and the other 2 nodes have out-degree 0.
Figure \ref{fig:transform} shows an example with $n=3$ leaves.
It is then evident that the space of heterochronous ranked tree shapes with $n$ leaves is
bijective with the subspace of isochronous ranked tree shapes with $2n$ leaves consisting of full-cherry trees.

\begin{figure}[h]
  \centering
  \includegraphics[width=0.8\linewidth]{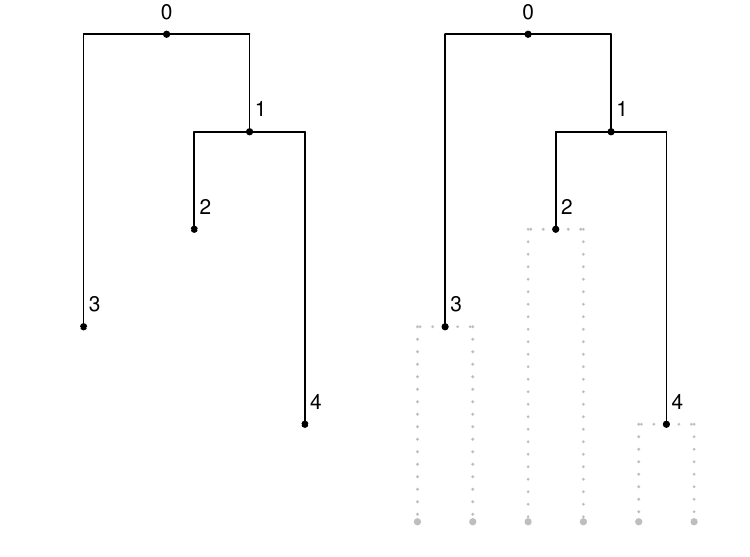}
  \caption{%
    A fully heterochronous ranked tree shape with 3 leaves (left) and  the corresponding full-cherry isochronous tree with 6 leaves and 3 cherries (right).
  }
  \label{fig:transform}
\end{figure}

Another consequence of the bijection between ranked tree shapes and full-cherry trees is that we can recursively count the number of fully heterochronous ranked tree shapes via a standard recursion involving root-splitting \citep[Chapter 2.1]{Steel-2016} to obtain the following proposition. 

\begin{prop}
  The number of fully heterochronous ranked tree shapes with $n$ leaves, $\left|\mathcal{T}^*_n\right|$, satisfies the following initial conditions and recursion,
  \begin{gather}\label{Gm_recursion}
    \left|\mathcal{T}^*_1\right| = \left|\mathcal{T}^*_2\right| = 1
    ,\qquad
	\left|\mathcal{T}^*_n\right| = \frac{1}{2} \sum_{\ell=1}^{n-1} \binom{2n-2}{2\ell-1} \left|\mathcal{T}^*_\ell\right| \left|\mathcal{T}^*_{n-\ell}\right|,
  \end{gather}
\end{prop}
\begin{proof}
  As the initial conditions are trivial, we assume $n\ge3$.
  To prove the recursion, we use that $\left|\mathcal{T}^*_n\right|$ is also the number of full-cherry trees with $2n$ leaves.
  The full-cherry trees with $2n$ leaves may be constructed as follows.
  Select two full-cherry trees $T_1$ and $T_2$, where $T_1$ has $2\ell$ leaves and $T_2$ has $2(n-\ell)$ leaves with $0<\ell<n$; extend the total orderings of the internal nodes of $T_1$ and $T_2$ to a common total ordering; join $T_1$ and $T_2$ with a new root node whose children are the roots of $T_1$ and $T_2$.
  Since $T_1$ has $2\ell-1$ internal nodes and $T_2$ has $2(n-\ell)-1$ internal nodes, there are exactly $\binom{2n-2}{2\ell-1}$ ways to extend to a common ordering.
  The sum in \eqref{Gm_recursion} corresponds to this construction, where the factor $\frac{1}{2}$  accounts for double counting due to constructing full-cherry trees from ordered pairs $(T_1,T_2)$ rather than unordered pairs $\{T_1,T_2\}$.
\end{proof}

We note that this recursion agrees with the recursion for strictly ordered binary trees by Poupard \cite{poupard1989}, although her proof does not involve full-cherry trees.
Indeed, Poupard's strictly ordered binary tree is a different name for the fully heterochronous ranked tree shape.
We restate the result to emphasize the value of working with a full-cherry tree, which will be the building block of the coalescent model introduced below. 
Using this recursion and an argument by generating functions, Poupard showed that the number of strictly ordered binary trees with $n$ leaves is equal to the $n^{th}$ reduced tangent number (see \eqref{eq.cardinality}). 
Table \ref{tab:cardinality} compares the cardinality of the space of isochronous ranked tree shapes and that of fully heterochronous ranked tree shapes for varying numbers of leaves. 

\begin{table}[h]
	\centering
	\scriptsize
	\begin{tabular}{l| cccccccccc}
		& \multicolumn{10}{c}{Number of leaves (n)}  \\		
		& $1$ & $2$ & $3$ & $4$ & $5$ & $6$ & $7$ & $8$ & $9$& $10$ \\
		\hline
		$\left|\mathcal{T}_n\right|$ & $1$ & $1$ & $1$ & $2$ & $5$ & $16$ & $61$ & $272$ & $1385$ & $7936$ \\
		$\left|\mathcal{T}^*_n\right|$ & $1$ & $1$ & $4$ & $34$ & $496$ &  $11056$ &  $349504$ & $14873104$ & $819786496$ & $56814228736$ \\ 
	\end{tabular}
	\caption{The number of isochronous and heterochronous ranked tree shapes with $n$ leaves, denoted respectively by $\left|\mathcal{T}_n\right|$ and $\left|\mathcal{T}^*_n\right|$.}
	\label{tab:cardinality}
\end{table}

In this section we introduce three methods for sampling fully heterochronous ranked tree shapes.
The first method is a coalescent model inspired by the bijection with full-cherry trees \citep[Proposition 2]{Palacios2015}.
This model is ``bottom-up'' in the sense that the generating process starts with one cherry node and adds cherries and merges cherries one by one until the root.
The second method, in contrast, is ``top-down'': it starts with the root and sequentially selects edges to bifurcate or to sample (terminate) as time moves forward.
This method utilizes the Catalan diagonal structure of the $\mathbf{F}$-matrix.
The last model generates one entry at a time sequentially along the $\mathbf{F}$-matrix via Bernoulli probabilities.
This last model can be specialized to a class of Beta-splitting models.

\subsection{Coalescent model} \label{subsection.coalescence_model}

The proposed coalescent model is a Markov chain whose full realization encodes a full-cherry tree, and therefore is an appropriate model for fully heterochronous ranked tree shapes (by removing the cherries at the end of the process).
The initial state is $2n$ leaves at the bottom of the tree.
We will describe the operation of connecting two nodes with a new node via two new edges as ``merging'' those two nodes.
The jump chain begins by forming a cherry, merging two leaves at a new node assigned rank $2n-2$.
To proceed, the chain introduces a new node and uniformly at random either merges two leaves or two non-leaf nodes at this new node.
The newly formed node is assigned a rank according to the time step when it was created, with older nodes assigned larger rank.
The $j$-th state of the chain is denoted by $A_{2n-j}=(L_{2n-j},V_{2n-j})$, where $L_{2n-j}$ denotes the number of nodes with total degree 0 (leaves not merged into cherries) and $V_{2n-j}$ denotes the set of ranks of non-leaf nodes with in-degree 0 (ranked nodes not merged) at step $j$.
The indices for states $A_{2n-j}$ run in reverse order compared to the steps $j$, which is standard for coalescent models.
By state $A_{2n-j}$, the Markov chain realizes a partially constructed full-cherry tree with nodes of ranks $2n-2$ to $2n-j$.
The chain starts at $A_{2n-1}=(2n,\emptyset)$ and completes after $2n-1$ steps at state $A_{0}=(0,\{0\})$ since the root is rank 0.
Figure \ref{fig:coal} shows an example.

With $k=2n-j$, the transition probability for state $j$ to $j+1$ is,
\begin{align*}
  &P(A_{k} \mid A_{k+1})
  \\
  &=
  \begin{cases}
    \dfrac{\binom{L_{k+1}}{2}}{\binom{L_{k+1}}{2}+\binom{\lvert V_{k+1}\rvert}{2}}
    &
    \text{if } L_{k}=L_{k+1}-2, V_{k+1}\subset V_k, \text{ and } \lvert V_{k}\setminus V_{k+1} \rvert =1,
    \\[2.5ex]
    \dfrac{1}{\binom{L_{k+1}}{2}+\binom{\lvert V_{k+1}\rvert}{2}}
    & \text{if } L_{k}=L_{k+1}, \lvert V_{k+1} \setminus V_{k}\rvert = 2,\text{ and } \lvert V_{k} \setminus V_{k+1} \rvert = 1,
    \\
    0 & \text{otherwise}.
  \end{cases}
\end{align*}

\begin{figure}[h!]
  \centering
  \includegraphics[width=0.8\textwidth]{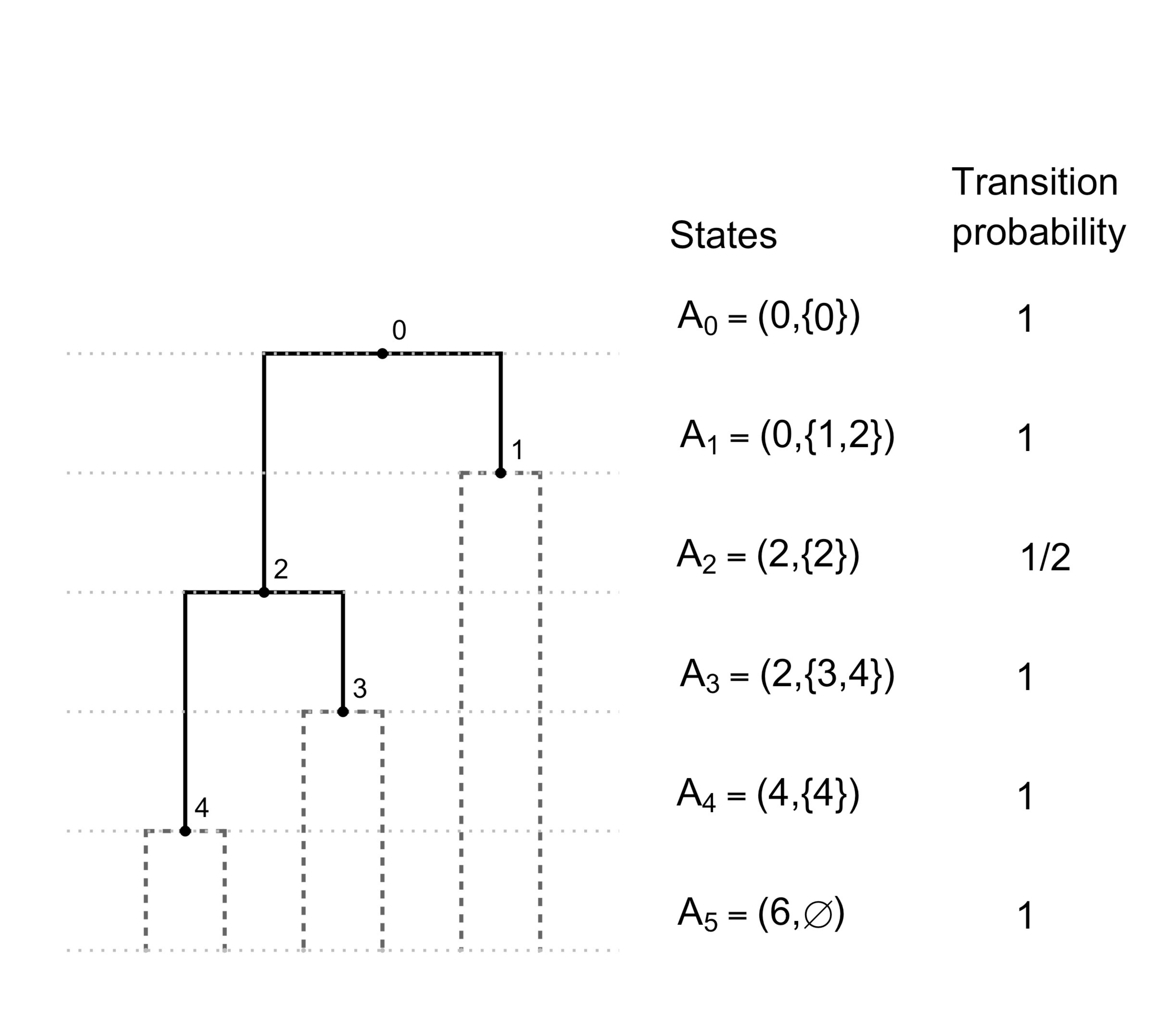}
  \caption{%
    An example of the coalescent jump chain with states $A_i$ and corresponding transition probabilities. The probability of the ranked tree shape is $\prod_{k=1}^3 P(A_k|A_{k+1}) = 1 \times \frac{1}{2} \times 1  = \frac{1}{2}$. 
  }
  \label{fig:coal}
\end{figure}

In order to compute the probability of a fully heterochronous ranked tree shape in terms of its $\mathbf{F}$-matrix $F$,
we need to determine the number of unmerged leaves and unmerged ranked nodes of the full-cherry tree from $F$.
For the number of unmerged ranked nodes, we have
\begin{gather*}
  |V_{2n-1}|=0, \qquad
  |V_k|=F_{k-1,k-1}\,\,\,\, \text{ for } 0<k<2n-1, \qquad
  V_{0}=1.
\end{gather*}
On the other hand, the total number of unmerged leaves and unmerged ranked nodes is $k+1$ at state $k$, and therefore the number of unmerged leaves is
\begin{gather*}
  L_{2n-1}=2n,\qquad
  L_{k} = k+1 - F_{k-1,k-1}\,\,\,\, \text{ for } 0<k<2n-1, \qquad
  L_0=0.
\end{gather*}

Therefore, the probability of a fully heterochronous ranked tree shape $T$ with $\mathbf{F}$-matrix $F$, under the coalescent model is:
\begin{equation}\label{eq.tree_prob_coal}
	\begin{aligned}
		P(T)
		&=
		\prod_{k=1}^{2n-3}  P(  A_k \mid A_{k+1} )
		\\
		&=
		\prod_{\substack{1\le k\le 2n-3,\\ F_{k,k}=F_{k-1,k-1}-1}} \binom{k+2-F_{k,k}}{2}
		\Bigg/
		{\displaystyle
			\prod_{0\le k\le 2n-3} \biggl\{\binom{k+2-F_{k,k}}{2} + \binom{F_{k,k}}{2}\biggl\}
		}
		.
	\end{aligned}
\end{equation}

\subsection{Diagonal ``top-down'' model} \label{subsection.diagonal_model}

A second model of fully heterochronous ranked tree shapes starts by uniformly generating the diagonal of the \textbf{F}-matrix (the sequence of coalescence and sampling events), and proceeds by uniformly at random selecting the edges for coalescence or sampling, conditioned on the matrix diagonal. To uniformly sample the diagonal, we rely on a bijection between the space of possible diagonal vectors and the space of Dyck paths from $(0,0)$ to point $(n-1,n-1)$.

\begin{figure}[h!]
  \centering
  \includegraphics[width=0.8\textwidth]{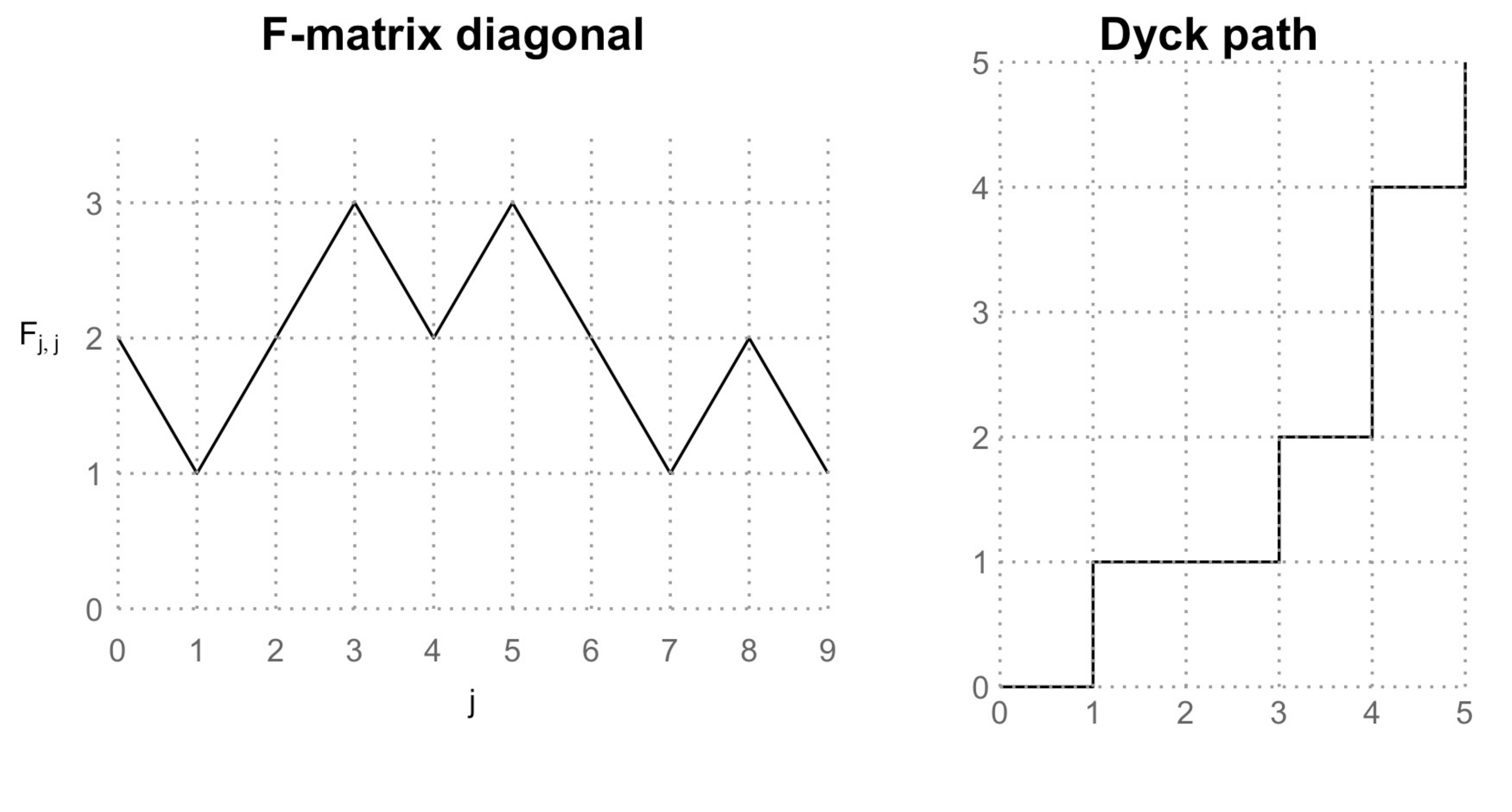}
  \caption{%
    An example of the \textbf{F}-matrix diagonal $[2, 1, 2, 3, 2, 3, 2, 1, 2, 1]$ and its corresponding Dyck path. Each unit decrease in the \textbf{F}-matrix diagonal is an upward step in the Dyck path and each unit increase is a rightward step.
  }
  \label{fig:diag_dyck}
\end{figure}

\begin{defn}
  A Dyck path is a path on the two-dimensional grid from point $(0,0)$ to point $(n-1,n-1)$ that can only move right or up by one unit, under the constraint that it never goes above the line $x=y$.
\end{defn}

\begin{prop} \label{prop:catalan}
  The number of possible diagonals in the \textbf{F}-matrix of a fully heterochronous ranked tree shape with $n$ leaves corresponds to the Catalan number $C_{n-1}$.
\end{prop}
\begin{proof}
  We first note that the diagonal of an \textbf{F}-matrix is equivalent to a Dyck path that starts at $(1,0)$ (corresponding to the initial 2 in the diagonal). Starting from the point $(1,0)$ in the Dyck path, if we record each rightward step as a $+1$ and each upward step as a $-1$, then we obtain a sequence of successive differences for a valid \textbf{F}-matrix diagonal that starts at $2$, ends at $1$, and takes only positive values. Hence the two spaces are bijective.

  It is well-known that the number of Dyck sequences of length $2n-2$ is the Catalan number $C_{n-1} = \frac{1}{n}\binom{2(n-1)}{n-1}$ \cite{Stanley-2012}. Therefore, the number of possible diagonals in the \textbf{F}-matrices of fully heterochronous trees with $n$ leaves is the Catalan number $C_{n-1}$.
\end{proof}

An algorithm to sample Dyck paths from $(1,0)$ to $(n-1,n-1)$ that has $O(n)$ complexity was proposed by \cite{Devroye1999}.
The algorithm proceeds sequentially starting from (1,0); at any point $(i,j)$ in the partially formed Dyck path, we move to the right with probability $N(i+1,j)/N(i,j)$, where
\begin{gather*}
  N(i,j) := \frac{i-j+1}{2n-1-i-j} \binom{2n-1-i-j}{n-j},
\end{gather*}
is the number of ways to complete the Dyck path from $(i,j)$ to $(n-1,n-1)$.
It is not hard to see that multiplying the transition probabilities results in a telescoping product equal to $1/N(0,1) = 1/C_{n-1}$.

Once the diagonal is sampled and fixed according to the previous algorithm, we have the order of bifurcation and sampling events of a tree.
For instance, if the diagonal is $[2, 3, 4, 3, 2, 1]$, then the sequence of successive differences is $[+1, +1, -1, -1, -1]$.
The tree has two bifurcations at times $u_1$ and $u_2$ and then three sampling events at $u_3$, $u_4$, and $u_5$.
Necessarily, $u_0$ is a bifurcation event and $u_6$ is a sampling event, so they are not included.

Next we need to sample the edges on which these events happen.
Generally at time $u_k$, with $1\le k\le 2n-2$, we have a partially constructed tree, and its corresponding partial \textbf{F}-matrix has $k$ complete rows.
We then choose an edge from the set of $F_{k-1,k-1}$ edges extant throughout $(u_k, u_{k-1})$ to be sampled or bifurcated at time $u_k$.
We label these extant edges with the rank of their parent node.
The number of such edges that descend from the node of rank $j$, with $j\le k-1$, is $D_{k-1,j} = F_{k-1,j}-F_{k-1,j-1}$.
Thus the probability of choosing an edge with rank label $L$, with $L\le k-1$, is
\begin{equation} \label{eq.prob_events_diag_samp}
  \frac{F_{k-1,L} - F_{k-1,L-1}}{F_{k-1,k-1}},
\end{equation}
If the chosen edge has label $L$, then the next row of the \textbf{F}-matrix (excluding diagonal) is given by
\[
  F_{k,j} =
  \begin{cases}
    F_{k-1,j} & j< L, \\
    F_{k-1,j} -1 & j \geq L.
  \end{cases}
\]

\begin{figure}[h!]
  \centering
  \includegraphics[width = 0.4\textwidth]{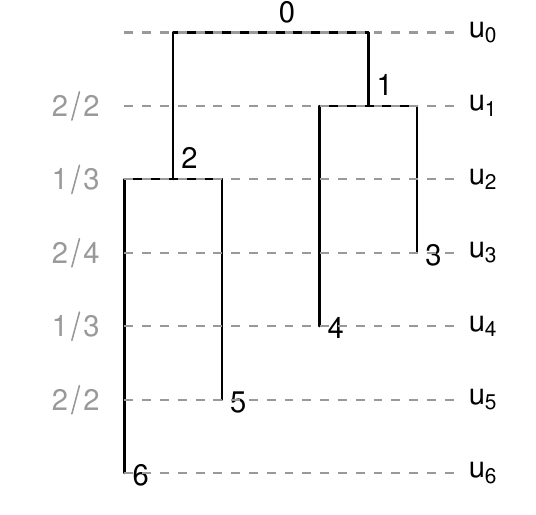}
  \caption{%
    An example of ``top-down'' sampled tree with fixed diagonal $[2, 3, 4, 3, 2, 1]$.
    Gray: probability of each bifurcation or sampling event that happens at respective times $u_1, \cdots, u_5$ according to \eqref{eq.prob_events_diag_samp}.
    The probability of the ranked tree shape is $\frac{2}{2} \times \frac{1}{3} \times \frac{2}{4} \times \frac{1}{3} \times \frac{2}{2} = \frac{1}{18}$.
  }
  \label{fig:ex_diag_samp}
\end{figure}

Continuing with the previous example of an $\mathbf{F}$-matrix with fixed diagonal $[2, 3, 4, 3, 2, 1]$, we can enumerate all $18$ compatible fully heterochronous ranked tree shapes (according to Proposition \ref{prop.extend_sequences}). Given that we randomly choose extant edges to sample or bifurcate, we would expect all \textbf{F}-matrices to be equally likely. We show that this is indeed the case.

\begin{prop} \label{prop.cond_uniform}
  The \textbf{F}-matrices conditioned on a fixed diagonal are uniformly distributed under the diagonal top-down model.
\end{prop}
\begin{proof}
  By \eqref{eq.prob_events_diag_samp}, the conditional probability is expressed in terms of $\mathbf{D}$ and $\mathbf{E}$-matrices as follows:
  \begin{gather*}
    P\left(F \mid \{F_{k,k}\}_{k=0}^{2n-3}\right)
    =
    \prod_{k=0}^{2n-3}\frac{F_{k,j_{k}} - F_{k,j_{k}-1} }{ F_{k,k}}
    =
    \frac{\prod_{k=0}^{2n-3}    D_{k,j_k}  }{\prod_{k=0}^{2n-3} F_{k,k}}
    ,
  \end{gather*}
  where $j_k = \argmax_{0\le j\le k} E_{k,j}$.
  Note $j_k$ is exactly the rank label of the edge selected for bifurcation or sampling at the event time $u_{k+1}$.
  The numerator $\prod_{k=0}^{2n-3}  D_{k,j_k}$ is the product of \textbf{D}-matrix entries that take values in $\{1,2\}$.
  In particular, $D_{k,j_k}$ is 2 when the sibling of the node of rank $k+1$ has rank larger than $k+1$, and is 1 when the sibling has rank smaller than $k+1$.
  Of all $2n-2$ non-root nodes, exactly half have rank larger than their sibling, so we have,

  \[
    P\left(F \mid \{F_{k,k}\}_{k=0}^{2n-3}\right) = \frac{2^{n-1 }}{\prod_{k=0}^{2n-3} F_{k,k}}.
  \]
  Notice that $P(F\mid \{F_{k,k}\}_{k=0}^{2n-3} )$ depends solely on the fixed diagonal $\{F_{k,k}\}_{k=0}^{2n-3}$, so all fully heterochronous ranked tree shapes with the same diagonal have the same probability. This concludes the proof.
\end{proof}

To summarize, the following proposition gives the probability of any fully heterochronous ranked tree shape under the diagonal top-down model.
\begin{prop} \label{prop.tree_prob_diag_samp}
  Under the diagonal top-down model, the probability of a fully heterochronous ranked tree shape with $\mathbf{F}$-matrix F is given by:
  \begin{equation} \label{eq.tree_prob_diag_samp}
    \begin{aligned}
      P(F) &= P\left(\{F_{k,k}\}_{k=0}^{2n-3}\right) P\left(F \mid \{F_{k,k}\}_{k=0}^{2n-3}\right) \\
      &= \frac{1}{C_{n-1}} \times
      \frac{2^{n-1}}{ \prod_{j=1}^{2n-2} F_{j-1,j-1}}.
    \end{aligned}
  \end{equation}
\end{prop}

\subsection{A Bernoulli splitting model} \label{subsection.bernoulli_model}

We now define a family of probability distributions on $\mathbf{F}$-matrices that sequentially generates one entry at a time conditioned on all previous values. Since each entry, conditioned on previous entries, can take up to two different values (see Theorem~\ref{theorem.heterochronous_f_matrix_tree}, constraint 2), these values can be sampled according to Bernoulli probabilities, except in trivial cases where the entry $F_{i,j}$ can take only a single value.

We further note that in determining valid values, we need at most four previous values rather than all previous values.
For a given non-trivial entry $F_{i,j}$, let $\mathcal{F}_{i,j} \mid  F_{<(i,j)}$ denote the set of possible values that $F_{i,j}$ can take conditionally of previous values, then
given real numbers $p_{i,j}\in(0,1)$, entry $F_{i,j}$ is $L_F(i,j)=\min(\mathcal{F}_{i,j}\mid F_{<(i,j)})$ with probability $p_{i,j}$ and $U_F(i,j)=\max(\mathcal{F}_{i,j}\mid F_{<(i,j)})$ with probability $1-p_{i,j}$.
We set
\begin{multline*}
  P(F_{i,j} \mid F_{<(i,j)}, p_{i,j})
  \\
  =
  \delta_{L_F(i,j)=U_F(i,j)}
  +
  (1-\delta_{L_F(i,j)=U_F(i,j)}) \, p_{i,j}^{\delta_{F_{i,j} = L_{F}(i,j)}}(1-p_{i,j})^{\delta_{F_{i,j}=U_F(i,j)}}, 
\end{multline*}
where $\delta$ is an indicator function. 

Let $\textbf{p}$ denote the set of Bernoulli probabilities $p_{i,j}$ with $(i,j)$ ranging over the non-trivial entries. The joint probability of an \textbf{F}-matrix conveniently telescopes as
\begin{gather*}
  P(F \mid \textbf{p})
  =
  \prod_{(i,j) \text{ non-trivial}} P(F_{i,j} \mid F_{<(i,j)}, p_{i,j})
  .
\end{gather*}

\begin{example}
  Figure \ref{fig:ex_enum1} shows the four \textbf{F}-matrices for fully heterochronous ranked tree shapes with three leaves, and marks (with $*$ sign) the three non-trivial positions. Recall the matrices are
  \begin{align*}
    F^0 =
    \begin{psmallmatrix}
      2 & 0 & 0 & 0\\
      1 & 1 & 0 & 0\\
      0 & 0 & 2 & 0\\
      0 & 0 & 1 & 1
    \end{psmallmatrix},
    F^1 =
    \begin{psmallmatrix}
      2 & 0 & 0 & 0\\
      1 & 3 & 0 & 0\\
      0 & 2 & 2 & 0\\
      0 & 1 & 1 & 1
    \end{psmallmatrix},
    F^2 =
    \begin{psmallmatrix}
      2 & 0 & 0 & 0\\
      1 & 3 & 0 & 0\\
      1 & 2 & 2 & 0\\
      0 & 1 & 1 & 1
    \end{psmallmatrix},
    F^3 =
    \begin{psmallmatrix}
      2 & 0 & 0 & 0\\
      1 & 3 & 0 & 0\\
      1 & 2 & 2 & 0\\
      1 & 1 & 1 & 1
    \end{psmallmatrix}.
  \end{align*}
  The probability of a matrix is:
  \begin{gather*}
    P(F \mid \textbf{p})
    =
    P(F_{1,1} \mid \textbf{p})
    P(F_{2,0} \mid F_{1, 1},\textbf{p})
    P(F_{3,0} \mid F_{2, 0},\textbf{p})
    P(F_{3,1} \mid F_{2, 0},F_{3,0},\textbf{p})
    .
  \end{gather*}
  By taking probabilities $p_{1,1}, p_{2,0}, p_{3,0}, p_{3,1}$, we find that
  \begin{align*}
    P(F^0\mid \textbf{p}) &= p_{1,1}
    ,
    &P(F^1\mid \textbf{p}) &= (1-p_{1,1})p_{2,0}
    ,\\
    P(F^2\mid \textbf{p}) &= (1-p_{1,1})(1-p_{2,0})p_{3,0}
    ,
    &P(F^3\mid \textbf{p}) &= (1-p_{1,1})(1-p_{2,0})(1-p_{3,0}),
  \end{align*}
  and $P(F^0\mid \textbf{p})+P(F^1\mid \textbf{p})+P(F^2\mid \textbf{p})+P(F^3\mid \textbf{p})=1$.
\end{example}

In the example above, the number of parameters $p_{i,j}$ and number of \textbf{F}-matrices are equal.
Additionally, the parameter $p_{3,1}$ is unnecessary.
These strange details are specific to the small number of leaves.
In general the number of parameters is much smaller, as the number of parameters is quadratic in $n$, whereas the number of matrices is comparable to $n^n$.

We highlight that the Bernoulli splitting model applies to the space of isochronous ranked tree shapes as well. For the isochronous ranked tree shapes, the diagonal and subdiagonal are fixed, and the remaining entries are chosen by a Bernoulli coin flip. Hence, the number of non-trivial entries is $(n-3)(n-2)/2$ in the isochronous case, compared to $2n^2-5n+1$ in the heterochronous case, where, as usual, $n$ is the number of leaves in the ranked tree shape.

Inspired by the Beta-splitting model \cite{aldous1996probability,sainudiin2016beta}, we can sample the Bernoulli probabilities from a Beta density $f(p_{i,j};\alpha,\beta)$, with parameters $\alpha$ and $\beta \in (0,\infty)$. The entry-wise probability is
\begin{align*}
  &P(F_{i,j} \mid F_{<(i,j)})
  \\
  &=
  \delta_{L_F(i,j)=U_F(i,j)}
  \\&\quad
  +
  (1-\delta_{L_F(i,j)=U_F(i,j)})\int^{1}_{0}p_{i,j}^{\delta_{F_{i,j} = L_{F}(i,j)}}(1-p_{i,j})^{\delta_{F_{i,j}=U_F(i,j)}} f(p_{i,j};\alpha,\beta)dp_{i,j}
  \\
  &=
  \delta_{L_F(i,j)=U_F(i,j)}
  \\&\quad
  +
  (1-\delta_{L_F(i,j)=U_F(i,j)})
  \frac{\text{B}(\alpha+\delta_{F_{i,j} = L_{F}(i,j)},\beta+ \delta_{F_{i,j} = L_{F}(i,j)})}{B(\alpha, \beta)} \\ 
  &= 
  \begin{cases}
  	1  & \text{if } L_F(i,j)=U_F(i,j)
  	,\\
  	\frac{\alpha}{\alpha+\beta} & \text{if } L_F(i,j)=F_{i,j} < U_F(i,j)
  	,\\ 
  	\frac{\beta}{\alpha+\beta} & \text{if }  L_F(i,j)< F_{i,j} = U_F(i,j)
  	.
  \end{cases} 
\end{align*}

Notice that the above equation implies that the marginal distribution of fully heterochronous ranked tree shapes has one parameter, which is the ratio $\frac{\alpha}{\beta}$.
This, however, should not be confused with the generative process, which is nonparametric.
In the generative model, we sample trees from the condition distribution of $P(F\mid \mathbf{p})$, where the number of parameters in $\mathbf{p}$ grows with tree size.

This Beta-Bernoulli model generates unbalanced trees when $\alpha\gg \beta$ and balanced otherwise. 
The mean of Beta($\alpha,\beta$) is $\frac{\alpha}{\alpha + \beta}$, so when $\alpha \gg \beta$, the Bernoulli probabilities of choosing the smaller admissible value $L_F$ (the probabilities of sampling as opposed to branching) tends to be higher.
Consequently, there are typically few extant lineages at any given time, which favors caterpillar-like shapes and yields unbalanced trees. 
This is seen in simulations, with the total tree length indicating how long and how many branches persist before sampling (Figure \ref{fig: sim_hist})

\begin{example}\label{ex.compare_tree_prob}
	To illustrate the differences between the three models introduced in this Section, we compute the probability of the ranked tree shape $T$  in Figure \ref{fig:coal} under each the three models. The tree has \textbf{F}-matrix,	
	\[
	F =
	\begin{pmatrix}
		2 & 0 & 0 & 0\\
		1 & 1 & 0 & 0\\
		0 & 0 & 2 & 0\\
		0 & 0 & 1 & 1
	\end{pmatrix}.
	\]
	By \eqref{eq.tree_prob_coal}, the probability under the Coalescent model is $P_{\text{coalescent}}(T) = \frac{1}{2}$, which agrees with the direct calculations in Figure \ref{fig:coal}. 
	Its probability under the Diagonal ``top-down'' model, according to Proposition \ref{prop.tree_prob_diag_samp}, is 
	\[P_{\text{diagonal}}(T) = \frac{2^2}{C_2 \times 2\times 1\times2\times1} = \frac{1}{2}.\]
	Lastly, for fixed Beta parameters $\alpha$ and $\beta$, the marginal tree probability is 
	\[P_{\text{Beta-Bernoulli}}(T) = \prod_{\text{nontrivial } (i,j)}P(F_{i,j} | F_{<(i,j)}) = P(F_{1,1} | F_{<(1,1)})   = \frac{\alpha}{\alpha+\beta} .\]
\end{example}

\subsection{Simulations}

We simulated $1000$ fully heterochronous ranked tree shapes with 5, 20, and 50 leaves, according to the three models defined in the previous sections.
For the Bernoulli splitting model, we simulated trees from three different Beta distributions: (1) $\alpha=10, \beta=1$, (2) $\alpha=10, \beta=10$, and (3) $\alpha=1, \beta=10$.

For each simulated tree, we computed 3 statistics:
the number of cherries, the total tree length (the sum of the number of $(u_i, u_{i+1})$ intervals that each branch survives), and the internal tree length (the sum of the number of $(u_i, u_{i+1})$ intervals that each internal edge survives).
Since our models are models on tree topology only, we assumed a unit length interval between consecutive events (branching or sampling).
The means of those statistics are presented in Tables~\ref{tab:sim.coales_diag} and \ref{tab:sim.bernoulli}.
Empirical distributions based on 1000 simulations of trees with 20 leaves are depicted in Figure~\ref{fig: sim_hist}.

\begin{table}[h]
  \centering
  \scriptsize
  \begin{tabular}{l| rrr| rrr}
    & \multicolumn{3}{c|}{coalescent} & \multicolumn{3}{c}{diagonal top-down} \\
    \hline
    $n$ & $5$ & $20$ & $50$ & $5$ & $20$ & $50$ \\
    \hline
    $N_C$ & 1.50 & 6.10 & 15.29 & 1.43 & 5.14 & 12.57 \\
    $L_{I}$ & 5.13 & 82.92 & 504.92 & 5.63 & 66.99 & 284.41 \\
    $L_{T}$ & 19.14 & 289.50 & 1787.50 & 17.28 & 155.04 & 617.56 \\
    \hline
  \end{tabular}
  \caption{%
    Comparing the average number of cherries $N_C$, the average internal tree length $L_I$, and the average total length $L_T$  of size-1000 samples of fully heterochronous ranked tree shapes (number of leaves $n=$ 5, 20, 50) from the coalescent model and the diagonal top-down model.
  }
  \label{tab:sim.coales_diag}
\end{table}

\begin{table}[h]
  \centering
  \scriptsize
  \begin{tabular}{l| rrr| rrr| rrr}
    & \multicolumn{3}{c|}{$\alpha=10,\beta=1$} & \multicolumn{3}{c|}{$\alpha=10,\beta=10$}& \multicolumn{3}{c}{$\alpha=1,\beta=10$} \\
    \hline
    $n$ & $5$ & $20$ & $50$ & $5$ & $20$ & $50$ & $5$ & $20$ & $50$ \\
    \hline
    $N_C$ & 1.01 & 1.08 & 1.44 & 1.39 & 5.12 & 12.20 & 1.37 & 6.27 & 19.12 \\
    $L_I$ & 5.97 & 36.01 & 96.48 & 5.37 & 61.16 & 264.61 & 3.64 & 62.43 & 521.16 \\
    $L_T$ & 12.45 & 59.97 & 157.09 & 17.18 & 157.92 & 607.47 & 23.30 & 379.57 & 2322.72 \\
    \hline
  \end{tabular}
  \caption{%
    Comparing the average number of cherries $N_C$, the average internal tree length $L_I$, and the average total length $L_T$ of size-1000 samples of fully heterochronous ranked tree shapes (number of leaves $n=$ 5, 20, 50) from the Bernoulli splitting model with different parameters for the Beta distribution ($\alpha = 10, \beta=1$; $\alpha = 10, \beta=10$; $\alpha = 1, \beta=10$).
  }
  \label{tab:sim.bernoulli}
\end{table}

Results from Tables \ref{tab:sim.coales_diag} and \ref{tab:sim.bernoulli} show that
among the two parameter-free models, the coalescent model generates samples with larger average internal length, total length, and number of cherries compared to the diagonal top-down model. On the other hand, by adjusting the hyperparameters of the beta distribution in the Bernoulli splitting model, the resulting sample can be quite different in terms of the three average statistics.

\begin{figure}[h!]
  \centerline{\includegraphics[width=1.25\textwidth]{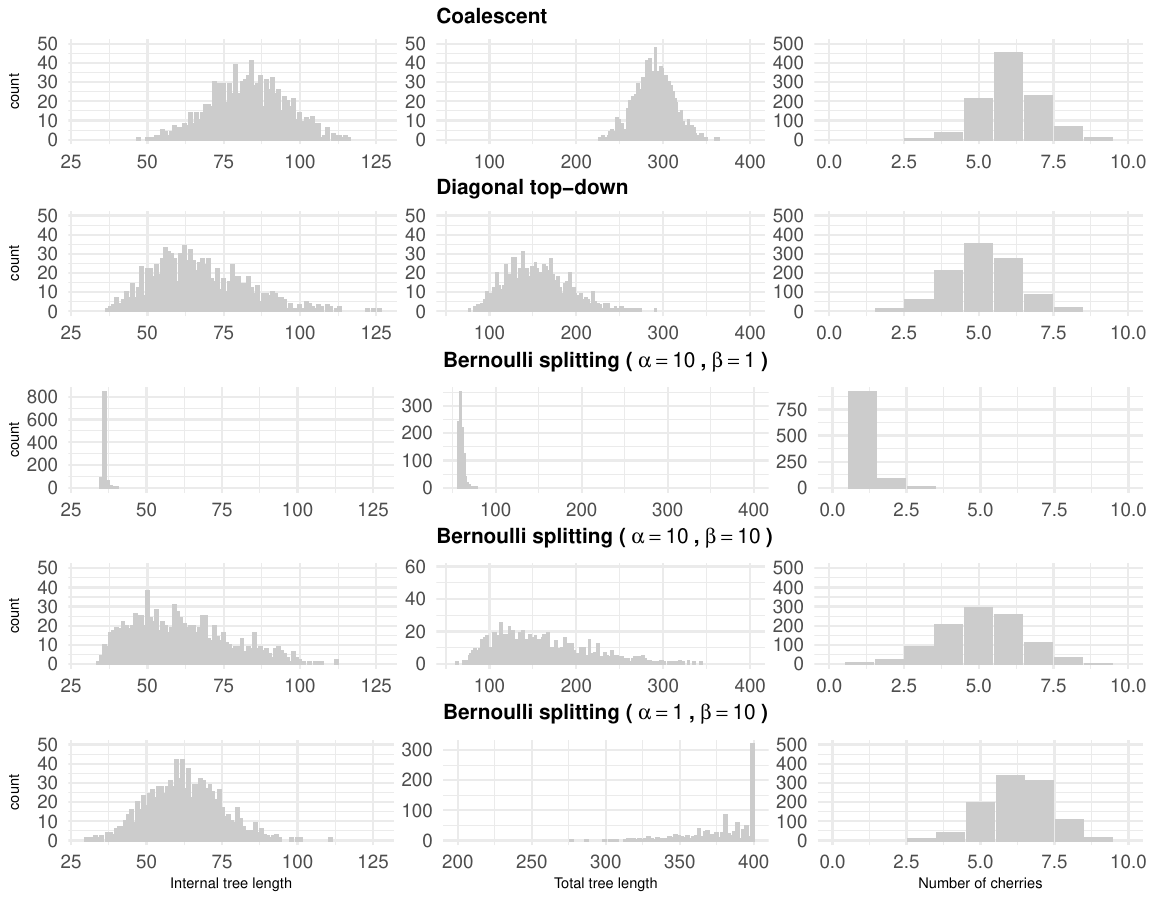}}
  \caption{%
    \small{Comparing sampling distributions of internal tree length, total tree length, and number of cherries for 1000 fully heterochronous ranked tree shapes with 20 leaves under the coalescent model, the diagonal top-down model, and the Bernoulli splitting model (with parameters $\alpha = 10, \beta=1$; $\alpha = 10, \beta=10$; $\alpha = 1, \beta=10$).}
  }
  \label{fig: sim_hist}
\end{figure}

The histograms in Figure \ref{fig: sim_hist} allow us to see the differences between samples more clearly.
The sampling distributions of summary statistics from the coalescent model and diagonal top-down model are roughly symmetric.
In contrast, the Bernoulli splitting model, regardless of hyperparameter values, produces more skewed distributions for total tree length.
The sampling distributions of the number of cherries are approximately symmetrical across all models.
As the ratio $\frac{\alpha}{\beta}$ decreases, the mode of the sampling distribution of total tree length increases, and the distributions change from being right-skewed to being left-skewed.
These simulations show that, even after we simplify the Bernoulli splitting model (so that the Bernoulli probabilities come from a Beta distribution), the model is very expressive in the sense that it can generate very different samples of trees.
We can thus reasonably conclude that, by adjusting the Bernoulli probabilities, the Bernoulli model can fit to various distributions.

\subsection{Implementation}\label{section.implementation}

Software implementing these methods is available in R and Python at \url{https://github.com/matsengrp/fully-heterochronous-f-matrix}.
It generates all \textbf{F}-matrices for trees of a given size, converts between \textbf{F}-matrix, \textbf{D}-matrix, and \textbf{E}-matrix formats, and validates ranked tree structures.
The enumeration algorithms use the characterizations from Section~\ref{section.theorems}.
The sampling implementations (Section~\ref{section.enumeration}) use the autoregressive structure of \textbf{F}-matrix construction.

\section{Discussion}\label{section.discussion}





In this article, we extended theorems describing \textbf{F}-matrices to fully heterochronous ranked tree shapes.
Using the \textbf{F}-matrix characterization, we were able to enumerate all fully heterochronous ranked tree shapes, and we highlighted our ability to construct \textbf{F}-matrices in an autoregressive order.
This construction allowed us to define a flexible family of probability distributions, with a large number of parameters, on the space of fully heterochronous ranked tree shapes. In addition, we introduced two parameter-free distributions that can serve as null distributions, which involve some uniform sampling at stages of tree formation. We then compared the flexible family of distributions against the two null distributions. Through simulations we showcased the ability of our flexible family to fit various and expressive distributions.

Note that we can attach the flexible family of probability distributions to isochronous ranked tree shapes. Additionally, the methods used here to characterize \textbf{F}-matrices for fully heterochronous ranked tree shapes can be applied to (non-fully) heterochronous ranked tree shapes with a fixed number of unique leaf sampling times. Likely one would need only to adjust the size of the matrix and conditions on the diagonal to account for the number of unique ranks. Then, we are equipped with representations and probability distributions on the entire space of ranked tree shapes.

In a future article, we will describe how to implement flexible probability distributions via neural networks.
Our goal is to model the distribution of tree shapes for B cell receptor sequences.
The present work has built a solid foundation: the probability associated with an entry of an \textbf{F}-matrix is written in terms of the probabilities of at most four previous entries.
This lends itself to an efficient autoregressive model.

\section*{Acknowledgments}

F.A.M.\ would like to thank Thierry Mora for discussions about the need for flexible models of tree shape that motivated the work in this paper.

This material is based upon work supported by the National Science Foundation under Grant No. DMS-1929284 while the senior authors met at the Institute for Computational and Experimental Research in Mathematics in Providence, RI, during the ``Algorithmic Advances and Implementation Challenges: Developing Practical Tools for Phylogenetic Inference'' program.

\section*{Funding}
J.A.P.\ acknowledges support from the NSF Career Award \#2143242 and NIH Award R35GM148338.
F.A.M.\ acknowledges support from NIAID award R01-AI146028.
Scientific Computing Infrastructure at Fred Hutch funded by ORIP grant S10OD028685.
Frederick Matsen is an investigator of the Howard Hughes Medical Institute.

\bibliographystyle{unsrt}
\bibliography{main.bib}

@ARTICLE{Gavryushkina2013-eq,
  title = "Recursive algorithms for phylogenetic tree counting",
  author = "Gavryushkina, Alexandra and Welch, David and Drummond, Alexei J",
  journal = "Algorithms Mol. Biol.",
  publisher = "Springer Science and Business Media LLC",
  volume =  8,
  number =  1,
  pages =  26,
  month =  "28~" # oct,
  year =  2013,
  url = "http://dx.doi.org/10.1186/1748-7188-8-26",
  doi = "10.1186/1748-7188-8-26",
  pmc = "PMC3829674",
  pmid =  24164709,
  issn = "1748-7188,1748-7188",
  language = "en"
}

@incollection{aldous1996probability,
  title={Probability distributions on cladograms},
  author={Aldous, David},
  booktitle={Random discrete structures},
  pages={1--18},
  year={1996},
  publisher={Springer}
}

@article{kingman1982coalescent,
  title={The coalescent},
  author={Kingman, John Frank Charles},
  journal={Stochastic processes and their applications},
  volume={13},
  number={3},
  pages={235--248},
  year={1982},
  publisher={Elsevier}
}

@ARTICLE{Stamatakis2014-xo,
  title = "{RAxML} version 8: a tool for phylogenetic analysis and post-analysis of large phylogenies",
  author = "Stamatakis, Alexandros",
  journal = "Bioinformatics",
  publisher = "Oxford Univ Press",
  volume =  30,
  number =  9,
  pages = "1312--1313",
  month =  may,
  year =  2014,
  url = "http://dx.doi.org/10.1093/bioinformatics/btu033",
  doi = "10.1093/bioinformatics/btu033",
  pmc = "PMC3998144",
  pmid =  24451623,
  issn = "1367-4803,1367-4811"
}

@article{sainudiin2016beta,
  title={A Beta-splitting model for evolutionary trees},
  author={Sainudiin, Raazesh and V{\'e}ber, Amandine},
  journal={Royal Society Open Science},
  volume={3},
  number={5},
  pages={160016},
  year={2016},
  publisher={The Royal Society},
  doi={10.1098/rsos.160016}
}

@ARTICLE{Minh2020-va,
  title = "{IQ}-{TREE} 2: New Models and Efficient Methods for Phylogenetic Inference in the Genomic Era",
  author = "Minh, Bui Quang and Schmidt, Heiko A and Chernomor, Olga and Schrempf, Dominik and Woodhams, Michael D and von Haeseler, Arndt and Lanfear, Robert",
  journal = "Mol. Biol. Evol.",
  volume =  37,
  number =  5,
  pages = "1530--1534",
  month =  may,
  year =  2020,
  url = "http://dx.doi.org/10.1093/molbev/msaa015",
  doi = "10.1093/molbev/msaa015",
  pmc = "PMC7182206",
  pmid =  32011700,
  issn = "0737-4038,1537-1719",
  language = "en"
}

@ARTICLE{Sagulenko2018-xl,
  title = "{TreeTime}: Maximum-likelihood phylodynamic analysis",
  author = "Sagulenko, Pavel and Puller, Vadim and Neher, Richard A",
  journal = "Virus Evol",
  publisher = "academic.oup.com",
  volume =  4,
  number =  1,
  pages = "vex042",
  month =  jan,
  year =  2018,
  url = "http://dx.doi.org/10.1093/ve/vex042",
  doi = "10.1093/ve/vex042",
  pmc = "PMC5758920",
  pmid =  29340210,
  issn = "2057-1577",
  language = "en"
}

@ARTICLE{Grenfell2004-dz,
  title = "Unifying the epidemiological and evolutionary dynamics of pathogens",
  author = "Grenfell, Bryan T and Pybus, Oliver G and Gog, Julia R and Wood, James L N and Daly, Janet M and Mumford, Jenny A and Holmes, Edward C",
  journal = "Science",
  volume =  303,
  number =  5656,
  pages = "327--332",
  month =  jan,
  year =  2004,
  url = "http://dx.doi.org/10.1126/science.1090727",
  doi = "10.1126/science.1090727",
  pmid =  14726583,
  issn = "0036-8075"
}

@ARTICLE{Drummond2007-co,
  title = "{BEAST}: Bayesian evolutionary analysis by sampling trees",
  author = "Drummond, Alexei J and Rambaut, Andrew",
  journal = "BMC Evol. Biol.",
  publisher = "biomedcentral.com",
  volume =  7,
  pages =  214,
  month =  nov,
  year =  2007,
  url = "http://dx.doi.org/10.1186/1471-2148-7-214",
  doi = "10.1186/1471-2148-7-214",
  pmc = "PMC2247476",
  pmid =  17996036,
  issn = "1471-2148"
}

@ARTICLE{Baele2025-fe,
  title = "{BEAST} {X} for Bayesian phylogenetic, phylogeographic and phylodynamic inference",
  author = "Baele, Guy and Ji, Xiang and Hassler, Gabriel W and McCrone, John T and Shao, Yucai and Zhang, Zhenyu and Holbrook, Andrew J and Lemey, Philippe and Drummond, Alexei J and Rambaut, Andrew and Suchard, Marc A",
  journal = "Nat. Methods",
  publisher = "Nature Publishing Group",
  pages = "1--4",
  month =  jul,
  year =  2025,
  url = "http://dx.doi.org/10.1038/s41592-025-02751-x",
  doi = "10.1038/s41592-025-02751-x",
  pmid =  40624354,
  issn = "1548-7091,1548-7105",
  language = "en"
}

@ARTICLE{Bouckaert2019-wv,
  title = "{BEAST} 2.5: An advanced software platform for Bayesian evolutionary analysis",
  author = "Bouckaert, Remco and Vaughan, Timothy G and Barido-Sottani, Joëlle and Duchêne, Sebastián and Fourment, Mathieu and Gavryushkina, Alexandra and Heled, Joseph and Jones, Graham and Kühnert, Denise and De Maio, Nicola and Matschiner, Michael and Mendes, Fábio K and Müller, Nicola F and Ogilvie, Huw A and du Plessis, Louis and Popinga, Alex and Rambaut, Andrew and Rasmussen, David and Siveroni, Igor and Suchard, Marc A and Wu, Chieh-Hsi and Xie, Dong and Zhang, Chi and Stadler, Tanja and Drummond, Alexei J",
  journal = "PLoS Comput. Biol.",
  volume =  15,
  number =  4,
  pages = "e1006650",
  month =  apr,
  year =  2019,
  url = "http://dx.doi.org/10.1371/journal.pcbi.1006650",
  doi = "10.1371/journal.pcbi.1006650",
  pmc = "PMC6472827",
  pmid =  30958812,
  issn = "1553-734X,1553-7358",
  language = "en"
}

@ARTICLE{Samyak2024-yo,
  title = "Statistical summaries of unlabelled evolutionary trees",
  author = "Samyak, Rajanala and Palacios, Julia A",
  journal = "Biometrika",
  publisher = "Oxford University Press (OUP)",
  volume =  111,
  number =  1,
  pages = "171--193",
  month =  mar,
  year =  2024,
  url = "https://academic.oup.com/biomet/article-pdf/111/1/171/56665918/asad025.pdf",
  doi = "10.1093/biomet/asad025",
  pmc = "PMC10861027",
  pmid =  38352626,
  issn = "0006-3444,1464-3510",
  language = "en"
}

@ARTICLE{Kim2020-ip,
  title = "Distance metrics for ranked evolutionary trees",
  author = "Kim, Jaehee and Rosenberg, Noah A and Palacios, Julia A",
  journal = "Proc. Natl. Acad. Sci. U. S. A.",
  publisher = "Proceedings of the National Academy of Sciences",
  volume =  117,
  number =  46,
  pages = "28876--28886",
  month =  nov,
  year =  2020,
  url = "https://www.pnas.org/doi/abs/10.1073/pnas.1922851117",
  doi = "10.1073/pnas.1922851117",
  pmc = "PMC7682335",
  pmid =  33139566,
  issn = "0027-8424,1091-6490",
  language = "en"
}

@book{Stanley-2012,
    author = "Stanley, Richard P",
    title = "Enumerative combinatorics, vol 1, 2nd edn",
    publisher = "Cambridge University Press",
    year = "2012"
}

@book{Steel-2016,
    author = "Steel, Michael A",
    title = "Phylogeny",
    publisher = "Society for Industrial and Applied Mathematics",
    year = "2016"
}

@misc{OEIS-hetero,
    key = "OEIS Foundation Inc. (2025)",
    note = "Entry A002105 in The On-Line Encyclopedia of Integer Sequences",
    url = "https://oeis.org/A002105"
}

@misc{OEIS-iso,
    key = "OEIS Foundation Inc. (2025)",
    note = "Entry A000111 in The On-Line Encyclopedia of Integer Sequences",
    url = "https://oeis.org/A000111"
}

@article{MooersHeard,
 ISSN = {00335770, 15397718},
 URL = {http://www.jstor.org/stable/3036810},
 author = {Arne Ø. Mooers and Stephen B. Heard},
 journal = {The Quarterly Review of Biology},
 number = {1},
 pages = {31--54},
 publisher = {The University of Chicago Press},
 title = {Inferring Evolutionary Process from Phylogenetic Tree Shape},
 urldate = {2025-05-19},
 volume = {72},
 year = {1997}
}

@article{Palacios2015,
	author  = {Palacios, Julia A. and Wakeley, John and Ramachandran, Sohini},
	title   = {Bayesian Nonparametric Inference of Population Size Changes from Sequential Genealogies},
	journal = {Genetics},
	volume  = {201},
	number  = {1},
	pages   = {281--304},
	year    = {2015},
	doi     = {10.1534/genetics.115.177980},
	url     = {https://doi.org/10.1534/genetics.115.177980}
}

@article{Devroye1999,
	author  = {Devroye, Luc and Flajolet, Philippe and Hurtado, Ferran and Noy, Marc and Steiger, William},
	title   = {Properties of Random Triangulations and Trees},
	journal = {Discrete \& Computational Geometry},
	volume  = {22},
	number  = {1},
	pages   = {105--117},
	year    = {1999},
	doi     = {10.1007/PL00009444},
	url     = {https://doi.org/10.1007/PL00009444}
}

@article{poupard1989,
title = {Deux Propriétés des Arbres Binaires Ordonnés Stricts},
journal = {European Journal of Combinatorics},
volume = {10},
number = {4},
pages = {369-374},
year = {1989},
issn = {0195-6698},
doi = {https://doi.org/10.1016/S0195-6698(89)80009-5},
url = {https://www.sciencedirect.com/science/article/pii/S0195669889800095},
author = {Christiane Poupard},
}

@article{donaghey1975,
title = {Alternating permutations and binary increasing trees},
journal = {Journal of Combinatorial Theory, Series A},
volume = {18},
number = {2},
pages = {141-148},
year = {1975},
issn = {0097-3165},
doi = {https://doi.org/10.1016/0097-3165(75)90002-3},
url = {https://www.sciencedirect.com/science/article/pii/0097316575900023},
author = {Robert Donaghey},
}

@techreport{foata1971,
author = {Dominique Foata, Marcel-Paul Schützenberger},
title = {Nombres d'Euler et permutations alternates},
institution = {University of Florida},
year = {1971}
}

@article{andre1881,
author = {Désiré André},
journal = {Journal de Mathématiques Pures et Appliquées},
language = {fre},
pages = {167-184},
title = {Sur les permutations alternées},
url = {http://eudml.org/doc/233984},
volume = {7},
year = {1881},
}

@ARTICLE{Victora2022-fm,
  title = "Germinal Centers",
  author = "Victora, Gabriel D and Nussenzweig, Michel C",
  journal = "Annu. Rev. Immunol.",
  month =  "3~" # feb,
  year =  2022,
  url = "http://dx.doi.org/10.1146/annurev-immunol-120419-022408",
  doi = "10.1146/annurev-immunol-120419-022408",
  pmid =  35113731,
  issn = "0732-0582,1545-3278",
  language = "en"
}

@ARTICLE{Tajima1983, 
	title ="Evolutionary relationship of DNA sequences in finite populations", 
	author = "Tajima, Fumio", 
	journal = "Genetics", 
	year = "1983", 
	volume = "105", 
	pages = "437-460", 
	doi = "10.1093/genetics/105.2.437", 
	pmid = 6628982
}

\end{document}